\RequirePackage{fix-cm}
\documentclass[smallextended]{svjour3}       
\smartqed  
\usepackage{graphicx}
\usepackage{mathptmx}      
\usepackage{cite}
\usepackage{amsmath}
\usepackage{amssymb}
\usepackage{amsfonts}
\usepackage{subfig}
\usepackage{commath}
\usepackage{hyperref}
\usepackage{bm}
\usepackage{mathtools}
\usepackage{tikz}
\usepackage{tikz-cd}
\usepackage[numbers]{natbib}

\newcommand{\p}{\mathbb{P}}

\newcommand{\dvr}{\widehat{\textnormal{DVR}}}
\newcommand{\dc}{\textnormal{D\v{C}}}

\newcommand{\e}{\mathbb{E}}

\newcommand{\g}{\tilde{g}}
\newcommand{\Mdim}{n}
\newcommand{\K}{\mathcal{K}}

\newcommand{\C}{\text{\v{C}}}
\newcommand{\diam}{\text{diam}}
\newcommand{\Hom}{\text{Hom}}
\newcommand{\dgm}{\text{dgm}}

\journalname{Foundations of Computational Mathematics}

\begin{document}

\title{A Family of Density-Scaled Filtered Complexes\thanks{The author acknowledges support from the Eugene V. Cota--Robles fellowship.}
}

\titlerunning{Density-Scaled Filtered Complexes}        

\author{Abigail Hickok
}


\institute{A. Hickok \at
              Department of Mathematics, University of California, Los Angeles, CA, USA \\
              \email{ahickok@math.ucla.edu}           
}


\maketitle

\begin{abstract}
We develop novel methods for using persistent homology to infer the homology of an unknown Riemannian manifold $(M, g)$ from a point cloud sampled from an arbitrary smooth probability density function. Standard distance-based filtered complexes, such as the \v{C}ech complex, often have trouble distinguishing noise from features that are simply small. We address this problem by defining a family of ``density-scaled filtered complexes'' that includes a density-scaled \v{C}ech complex and a density-scaled Vietoris--Rips complex. We show that the density-scaled \v{C}ech complex is homotopy-equivalent to $M$ for filtration values in an interval whose starting point converges to $0$ in probability as the number of points $N \to \infty$ and whose ending point approaches infinity as $N \to \infty$. By contrast, the standard \v{C}ech complex may only be homotopy-equivalent to $M$ for a very small range of filtration values. The density-scaled filtered complexes also have the property that they are invariant under conformal transformations, such as scaling. We implement a filtered complex $\dvr$ that approximates the density-scaled Vietoris--Rips complex, and we empirically test the performance of our implementation. As examples, we use $\dvr$ to identify clusters that have different densities, and we apply $\dvr$ to a time-delay embedding of the Lorenz dynamical system. Our implementation is stable (under conditions that are almost surely satisfied) and designed to handle outliers in the point cloud that do not lie on $M$.


\end{abstract}


\section{Introduction}
Data in Euclidean space often lie on (or near) a lower-dimensional submanifold $M$. For example, images with many pixels are high-dimensional, but image libraries are often locally parameterized by many fewer dimensions \cite{isomap}. In chemistry, the conformation space of a molecule may be a manifold or a union of manifolds \cite{cyclo}. In \emph{topological data analysis} (TDA), one considers the following question: given a finite sample of points (a \emph{point cloud}) that lies on or near $M$, what can one infer about the topology (i.e., ``global structure'') of $M$? TDA has been used to study the global structure of data sets in a variety of fields (see, e.g., \cite{cortical, topaz, materials}). Researchers have also made significant progress towards using the geometric properties of the manifold for dimensionality reduction and data visualization \cite{isomap, laplacian_eigmap, hessian_eigenmap, LLE}.

We focus on inferring the \emph{homology} of $M$. Homology is a quantitative way of characterizing the topology of $M$. For example, the rank of the $0$-dimensional homology $H_0(M)$ is the number of connected components, and the rank of $H_k(M)$ for $k \geq 1$ is the number of $k$-dimensional holes in $M$. If $M$ is compact and orientable, then the dimension of $M$ is equal to the largest $n$ such that $H_n(M)$ is nontrivial. For example, if $M$ is the $2$-torus $\mathbb{S}^1 \times \mathbb{S}^1$, then there is one connected component, there are two $1$-dimensional holes, and there is one $2$-dimensional hole. Although homology does not uniquely identify a manifold, it provides useful information about a manifold's global structure, and the homology of a manifold can be used to distinguish it from other manifolds that have different homology.

Methods from \emph{persistent homology} (PH) can be used to infer the homology of $M$ from a point cloud $X = \{x_i\}_{i=1}^N$ that is sampled from $M$. To approximate the manifold, we construct a \emph{filtered complex}, a combinatorial description of a topological space (see Definition \ref{def:FSC}). One of the classical approaches to building a filtered complex is the \emph{\v{C}ech complex} $\C(X)$. At each point $x \in X$, one places a ball of radius $r > 0$, where $r$ is the \emph{filtration level}. A $k$-simplex with vertices $\{x_{i_j}\}$ is added to $\C(X)_r$ if the intersection $\bigcap_j B(x_{i_j}, r)$ is nonempty. The Nerve Theorem guarantees that $\C(X)_r$ is homotopy-equivalent to $\bigcup_i B(x_i, r)$. The PH of $\C(X)$, which we denote by $H(\C(X))$, records how the homology of $\C(X)_r$ changes as $r$ increases. As $r$ grows, new homology classes (which represent $k$-dimensional holes) are ``born'' and old homology classes ``die''.

Conventional wisdom holds that the homology classes with the longest lifetimes are true topological features of $M$ and that the homology classes with the shortest lifetimes are noise. However, one can easily observe that this is not always true, even in simple examples such as Figure \ref{fig:twocircles_cech}, in which the point cloud is sampled from the disjoint union of two circles of different sizes. The smaller circle represents a homology class that has a much shorter lifetime than the homology class for the larger circle, but both homology classes are true topological features. We visualize this in Figure \ref{fig:twocircles_cech}, in which the balls of the \v{C}ech complex fill in the smaller circle much earlier than they fill in the larger circle.
\begin{figure}
    \centering
    \subfloat[$r =0$]{\includegraphics[width = .33\textwidth]{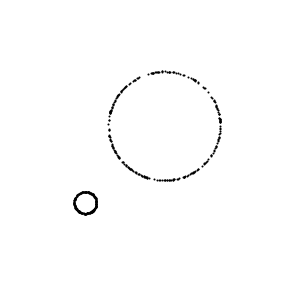}\label{fig:twocircles}}
    \subfloat[$r =1$]{\includegraphics[width = .33\textwidth]{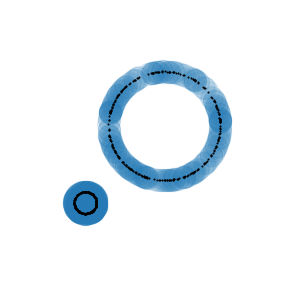}}
    \subfloat[$r = 2$]{\includegraphics[width = .33\textwidth]{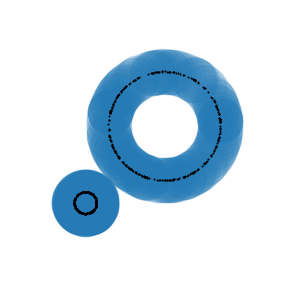}} \\
    \subfloat[$r = 3$]{\includegraphics[width = .33\textwidth]{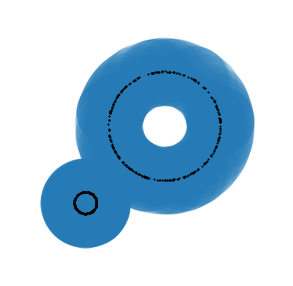}}
    \subfloat[$r = 4$]{\includegraphics[width = .33\textwidth]{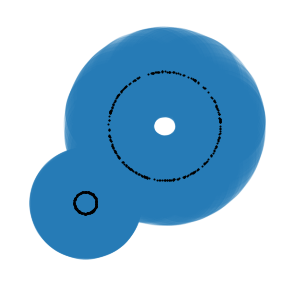}}
    \subfloat[$r = 5$]{\includegraphics[width = .33\textwidth]{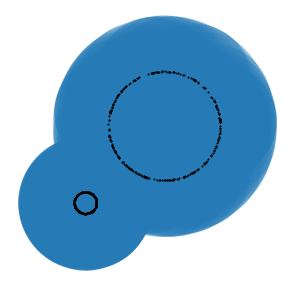}}
    \caption{The point cloud $X$ consists of $N = 500$ points that are sampled from the disjoint union of two circles $C_1$ and $C_2$ with radii $R_1 = 1$ and $R_2 = 5$. With probability $1/2$ we sample uniformly at random from $C_1$, and with probability $1/2$ we sample uniformly at random from $C_2$. For increasing $r$, we display the balls of radius $r$ in the \v{C}ech complex $\C(X)_r$ at filtration level $r$. At $r = 0$, we have the point cloud itself. The smaller circle is filled in immediately at the next step, $r = 1$, but the larger circle is not filled in until $r = 5$.}
    \label{fig:twocircles_cech}
\end{figure}
Following the conventional wisdom, the homology class for the smaller circle might be recorded spuriously as noise. Problems with the conventional wisdom have been noted in many other papers, such as \cite{roadmap, pers_images, bendich2016, stolz2017, Bubenik2020, feng2021}.

In general, standard distance-based filtered complexes (such as the \v{C}ech complex) depend largely on ``topological feature sizes,'' by which we mean the following concept, introduced in \cite{feature_size}. The \emph{medial axis} of a submanifold $M$ in $\mathbb{R}^m$ is the closure of
\begin{equation*}
    G := \{x \in \mathbb{R}^m \mid \text{there are distinct } y, z \in M \text{ such that } d(x, M) = d(x, y) = d(x, z)\}.
\end{equation*}
The \emph{local feature size} at $x \in M$, denoted by $\sigma(x)$, is the distance from $x$ to the medial axis. The \emph{condition number} of $M$ is equal to $1/\tau$, where $\tau = \inf_{x \in M} \sigma(x)$. For example, if $M$ is an $n$-sphere in $\mathbb{R}^{n+1}$, then the medial axis is the center of the sphere and the local feature size at any point on the sphere is the radius of the sphere. Niyogi et al. showed that $\C(X)_r$ is homotopy-equivalent to $M$ when $r < \sqrt{\frac{3}{5}} \tau$ and $X$ is sufficiently dense in $M$ \cite{weinberger}. However, whenever $\tau$ is small, the \v{C}ech complex may only be homotopy-equivalent to $M$ for a very small range of filtration values $r$, even as the number of points sampled from the manifold approaches infinity.

Standard distance-based filtered complexes may perform especially poorly when $M$ contains features of different sizes, even if the smallest features have ``high resolution'' in the point cloud (i.e., the density of points is inversely proportional to the local feature size). For example, consider again the point cloud in Figure \ref{fig:twocircles}, sampled from the disjoint union of two circles $C_1$ and $C_2$ of different radii. (With probability $1/2$ we sample uniformly at random from $C_1$, and with probability $1/2$ we sample uniformly at random from $C_2$.) The product of the probability density function and the local feature size is a constant function; in that sense, the two circles have equally high resolution. However, the corresponding homology classes do not have equally high persistence in the PH of standard filtered complexes.

The dependence on topological feature size is because persistent homology is not a topological invariant. The topology of a manifold is invariant under homeomorphism, but standard distance-based filtered complexes (such as the \v{C}ech complex) are not invariant under homeomorphism. More precisely, suppose $F: M \to M'$ is a homeomorphism of manifolds and $X$ is a point cloud in $M$. The manifolds $M$ and $M'$ are homeomorphic, but $\C(X)$ and $\C(F(X))$ are not necessarily isomorphic (see Definition \ref{def:isomorphism}). Indeed, the bottleneck distance between the persistence diagrams (see Section \ref{sec:pers_modules}) for $H(\C(X))$ and $H(\C(F(X))$ can be arbitrarily large\footnote{For example, consider the scaling homeomorphism $F: \mathbb{R}^m \to \mathbb{R}^m$ defined by $F(\bm{x}) = \lambda \bm{x}$ for some $\lambda \in (0, \infty)$. For any point cloud $X$ with more than one point, the bottleneck distance between $H(\C(X))$ and $H(\C(F(X))$ approaches infinity as $\lambda \to \infty$.}. Therefore, the standard \v{C}ech complex depends on geometric properties such as size. Standard distance-based filtered complexes are closer to geometric tools than topological tools.

\subsection{Contributions}
We work in a probabilistic setting. We suppose that $(M, g)$ is an $n$-dimensional Riemannian manifold and that the point cloud $X$ consists of $N$ points sampled from a smooth probability density function $f: M \to (0, \infty)$. It is important that $f$ is nonzero everywhere because we cannot observe regions of the manifold where $f(x)$ equals zero. The Riemannian metric is necessary because it turns the manifold $M$ into a metric space and induces a volume form $dV$. We define the probability measure to be $\p[A] = \int_A f dV$, where $A \subseteq M$ is a Borel set \cite{Rman_stats}. We note that all manifolds can be endowed with a Riemannian metric (see Section \ref{sec:Rgeom}), so the requirement of a Riemannian metric is not a restriction on the types of manifolds we can study.

We construct a family of ``density-scaled filtered complexes'' by modifying the metric $g$ such that we effectively shrink the distances between points in sparse regions of the manifold and enlarge the distances between points in dense regions of the manifold. To do this, we define a conformally equivalent metric $\g := \sqrt[n]{f^2\alpha(N)^2} g$, where $\alpha(N)$ is a scaling factor that we define in Section \ref{sec:def_Mscaled}. 
Our scaling factor $\alpha(N)$ plays an important role in the convergence property that we prove in Section \ref{sec:convergence} and discuss below. The metric $\g$ is defined such that the points in $X$ are uniformly distributed with respect to the volume form $d\widetilde{V}$ in $(M, \g)$ and such that the balls grow at a slower rate when $N$ is larger. We can then apply any existing distance-based filtered complex (such as the \v{C}ech complex) in the density-scaled Riemannian manifold $(M, \g)$.

We show that our density-scaled filtered complexes have two important properties that other filtered complexes do not have:
\begin{enumerate}
    \item Convergence: As $N \to \infty$, the interval of filtration values for which the density-scaled \v{C}ech complex is homotopy-equivalent to the manifold $M$ grows to $(0, \infty)$ in probability, no matter the condition number of $M$ or any other geometric property of $M$. (We make this statement precise in Theorem \ref{thm:convergence}.) This means that in the PH of the density-scaled \v{C}ech complex, one can interpret the homology classes with the smallest birth times and longest lifetimes as the most important features.
    \item Conformal invariance: We show that our density-scaled filtered complexes are invariant under conformal transformations (Theorem \ref{thm:local_invar}). This means that in contrast to standard complexes, our density-scaled complexes are closer to topological tools and do not depend as much on local feature sizes.
\end{enumerate}
These properties improve our ability to infer the homology of $M$ from a point cloud and make it easier to compare the PH of point clouds sampled from different manifolds of possibly different scales.

We implement a filtered complex $\dvr$ that approximates the density-scaled Vietoris--Rips complex. We do this by estimating the density $f$ via kernel-density estimation and estimating Riemannian distances in a similar way as the widely-used Isomap algorithm \cite{isomap}. The implementation requires knowledge of the intrinsic dimension $n$ of the manifold, which can be estimated using methods such as local principal component analysis \cite{local_PCA, intrinsic_dim}, the conical dimension estimator \cite{conical_dimestimate}, the ball expansion rate \cite{ball_dimestimate}, or the doubling dimension \cite{doubling_dimestimate}. 
We prove that our implementation $\dvr$ is stable (Theorem \ref{thm:dvr_stability}): under suitable conditions that are almost surely satisfied, small perturbations of the input point cloud $X$ result only in small changes to the persistence diagram of $\dvr(X)$.
Consequently, it is still reasonable to use $\dvr$ even when $X$ does not lie exactly on the manifold $M$ or when there is a small amount of noise in the data. The implementation is designed to handle outliers in the data; in Section \ref{sec:geo_est} we discuss how this is done, and in Section \ref{sec:outliers} we test the empirical performance of $\dvr$ on a point cloud with outliers. As applications, we use $\dvr$ to count the number of clusters in a point cloud whose clusters have different densities (Section \ref{sec:clustering}) and the number of equilibrium points in the Lorenz dynamical system from a time-delay embedding (Section \ref{sec:lorenz}).


\subsection{Related Work}
Perhaps the most common TDA-based approach to nonuniform data is the $k$-nearest neighbor (KNN) filtration (see Appendix \ref{sec:knn}). This is related to the density-scaled filtrations by the fact that if $x_{N^k_i}$ is the $k$th nearest neighbor of $x_i$, then $\norm{x_i - x_{N^k_i}}$ converges in probability to a value that is proportional to $f(x_i)^{-1/n}$ as $N \to \infty$, for a choice of $k$ that depends on $N$. (See \cite{knn_density} for a precise statement.)
However, the KNN filtration encounters problems when there are regions of the manifold that are close in Euclidean distance but far in Riemannian distance, especially if those regions differ in density. We discuss one example in Section \ref{sec:clustering}; several other examples of KNN failures are given in \cite{continuous_knn}. In \cite{continuous_knn}, Berry and Sauer constructed a modification of the $k$-nearest neighbors graph (the \emph{continuous $k$-nearest neighbors graph}) whose unnormalized graph Laplacian converges to the Laplace--Beltrami operator of a slightly different density-scaled Riemannian manifold. (Their density-scaled metric is $\sqrt[n]{f^2}g$, where $g$ is the original metric.) The authors of \cite{continuous_knn} proved that the connected components of their graph were consistent with the components of the manifold. They left as conjecture the hypothesis that their graph was topologically consistent (i.e., that the $k$-dimensional homology of the clique complex of their graph converges to the $k$-dimensional homology of the manifold for $k > 0$).

A qualitatively different family of density-scaled metrics was considered in \cite{fermat_tda}. For parameter $p > 1$, the density-scaled metric in \cite{fermat_tda} is $\frac{1}{\sqrt[n]{f^{2(p-1)}}}g$. The Riemannian distance induced by the density-scaled metric of \cite{fermat_tda} is called the \emph{Fermat distance} \cite{fermat}. The Fermat distance effectively enlarges the distances between points in sparse regions of the manifold and shrinks the distances between points in dense regions of the manifold; by contrast, the density-scaled metric in the present paper does the opposite.

The density-scaled complexes in the present paper are also reminiscent of weighted complexes \cite{weightedPH}. (See Appendix \ref{sec:weighted} for a review of weighted complexes.) In a weighted \v{C}ech complex, the radius of a ball is a function $r_x(t)$ of the filtration parameter $t$ and the point $x$ at which the ball is centered\footnote{The radius function $r_x(t)$ need not depend on density; more typically, the weight is determined by some intrinsic property of the point. For example, in \cite{weightedPH}, a point cloud that represented the positions of image pixels had weights that were given by pixel intensity.}. Weighted Vietoris--Rips complexes are defined analogously. One can define a ``density-weighted'' radius function
\begin{equation}\label{eq:weighted_radius}
    r_x(t) := \frac{t}{\sqrt[n]{\alpha(N)f(x)}}
\end{equation}
from which one can define a density-weighted \v{C}ech complex and a density-weighted Vietoris--Rips complex.
The main advantage of our density-scaled complexes over the density-weighted complexes is that our complexes are more robust with respect to noise. Specifically, if $x \in X$ is an outlier in a low-density region, then the ball $B(x, r_x(t))$ grows quickly in radius and may engulf balls in high-density regions. This problem can occur even if all the data points $x$ lie exactly on the manifold $M$. If $A \subseteq M$ is a high-density region, then balls centered at points $x \in A$ grow quickly in radius and may engulf points in low-density regions of $M$. In Sections \ref{sec:outliers} and \ref{sec:clustering}, we calculate examples and discuss these problems in more detail.

Other density-based filtrations, such as the distance-to-measure (DTM) sublevel filtration \cite{dtm_tda} and the density sublevel filtration \cite{kde_sublevel}, are primarily designed for the purpose of noise filtering. Such methods assume that the regions of highest density are the true features of the manifold. For example, consider the point cloud of Figure \ref{fig:twocircles} again. In these other density-based filtrations, it is the \emph{smaller} circle whose corresponding homology class has a much longer lifetime in the persistent homology. In our density-scaled filtration, the two circles have equal lifetimes in the persistent homology, which reflects the fact that they have equally high ``resolution'' in the point cloud.

\subsection{Organization}
The rest of the paper is organized as follows. In Section \ref{sec:background}, we review background from TDA and Riemannian geometry. In Section \ref{sec:defs}, we introduce our family of density-scaled filtered complexes, including definitions for a density-scaled \v{C}ech complex ($\dc$) and a density-scaled Vietoris--Rips complex (DVR). We discuss convergence properties in Section \ref{sec:convergence} and invariance properties in Section \ref{sec:invariance}. In Section \ref{sec:implement}, we discuss our algorithm for the implementation of a filtered complex $\dvr$ that approximates the density-scaled Vietoris--Rips complex. We prove the stability of our density-scaled complexes (including a stability theorem for $\dvr$) in Section \ref{sec:stability}. In Section \ref{sec:compare}, we compute examples and compare $\dvr$ to other filtered complexes. Finally in Section \ref{sec:discussion}, we conclude and discuss some avenues for future research. The code used in this paper is available at \url{https://bitbucket.org/ahickok/dvr/src/main/}.


\section{Background}\label{sec:background}

\subsection{Filtered Complexes}

A comprehensive introduction to filtered complexes and TDA can be found in \cite{edel_book, eat}. Here we review the standard methods for building a filtered complex. Throughout this section, let $(M, d)$ denote a metric space and let $X = \{x_1, \ldots, x_N\}$ denote a point cloud in $M$. For any index set $J \subseteq \{1, \ldots, N\}$, let $x_J$ denote the simplex with vertices $x_j$ for all $j \in J$. 

\begin{definition}\label{def:FSC}
A \emph{filtered complex} $\K$ is a collection of simplicial complexes $\{\K_r\}_{r \in \mathbb{R}}$ such that $\K_s \subseteq \K_r$ for all $s \leq r$. We refer to $r$ as the \emph{filtration level}.
\end{definition}

\begin{definition}
The \emph{\v{C}ech complex} $\C(M, d, X)$ is the filtered complex such that the set of simplices in $\C(M, d, X)_r$ at filtration level $r$ is 
\begin{equation*}
    \Big\{ x_J \mid \bigcap_{j \in J} B(x_j, r) \neq \emptyset \text{ and }J \subseteq \{1, \ldots, N\} \Big\}.
\end{equation*}
Equivalently, $\C(M, d, X)_r$ is the nerve of $\{B(x, r)\}_{x \in X}$, where $B(x, r) := \{y \in M \mid d(x, y) \leq r\}$.
\end{definition}

The Nerve Theorem provides theoretical guarantees for the \v{C}ech complex \cite{Borsuk}.
\begin{theorem}[Nerve Theorem]\label{thm:nerve}
If $\bigcap_{j \in J} B(x_j, r)$ is either contractible or empty for all $J \subseteq \{1, \ldots, N \}$, then $\C(M, d, X)_r$ is homotopy-equivalent to $\bigcup_{i=1}^N B(x_i, r)$.
\end{theorem}
In Euclidean space, all balls are convex (hence their intersections are contractible), and thus the \v{C}ech complex at filtration level $r$ is homotopy-equivalent to $\bigcup_{x \in X} B(x, r)$. In an arbitrary metric space, however, balls are not always convex. In a Riemannian manifold, $\bigcap_{j \in J}B(x_j, r)$ is contractible only when $r$ is sufficiently small.

Computing the \v{C}ech complex is computationally intensive. In practice, researchers often compute the Vietoris--Rips complex instead, which requires only pairwise distances between the points.
\begin{definition}
The \emph{Vietoris--Rips complex} $\textnormal{VR}(M, d, X)$ is the filtered complex such that the set of simplices in $\textnormal{VR}(M, d, X)_r$ at filtration level $r$ is
\begin{equation*}
    \Big\{x_J \mid d(x_i, x_j) \leq 2r \text{ for all }i, j \in J\, \text{ and } J \subseteq \{1, \ldots, N\} \Big\}.
\end{equation*}
\end{definition}
The Vietoris--Rips complex and the \v{C}ech complex share the same 1-skeleton. When the metric space $(M, d)$ is Euclidean space, the Vietoris--Rips complex and the \v{C}ech complex are related by the Vietoris--Rips lemma \cite{edel_book}, which says that
\begin{equation*}
    \C(M, d, X)_r \subseteq \textnormal{VR}(M, d, X)_r \subseteq \C(M, d, X)_{\sqrt{2}r}
\end{equation*}
for all filtration values $r$. In addition to the \v{C}ech and Vietoris--Rips complexes, there are many other methods for constructing a filtered complex from a point cloud. We review other relevant filtered complexes in Appendix \ref{sec:otherfilts}.


\subsection{Persistence Modules}\label{sec:pers_modules}

In this section, we define persistence modules, persistent homology, and persistence diagrams. We assume the reader is familiar with homology. (A good introduction to homology and algebraic topology is \cite{hatcher}.) References for the rest of this subsection can be found in \cite{fundthm, pers_modules}.

A \emph{persistence module} $\mathbb{V}$ over $\mathbb{R}$ is a collection of vector spaces $\{\mathbb{V}_t\}_{t \in \mathbb{R}}$ with linear maps $\{v_s^t : \mathbb{V}_s \to \mathbb{V}_t \text{ for all } s \leq t\}$ that satisfy the composition law $v_t^u \circ v_s^t = v_s^u$ for all $s \leq t \leq u$. If $\K$ is a filtered complex, the \emph{persistent homology} of $\K$ over a field $\mathbb{F}$ is the persistence module $\{H(\K_r, \mathbb{F})\}_{r \in \mathbb{R}}$, which we denote by $H(\K, \mathbb{F})$. For all $s \leq t$, the inclusion $\K_s \xhookrightarrow{} \K_t$ induces a linear map $\iota_s^t : H(\K_s, \mathbb{F}) \to H(\K_t, \mathbb{F})$. We sometimes drop the field $\mathbb{F}$ from our notation when a fixed field is chosen. (All calculations in Section \ref{sec:compare} are done with $\mathbb{F} = \mathbb{Z}/11\mathbb{Z}$, the default field used by the GUDHI software package.) As $r$ increases, new homology classes are born and old homology classes die.

The Fundamental Theorem of Persistent Homology, stated below, shows that we can decompose the persistence module in a way that yields a nice set of generators. If $\K_t \setminus \K_s$ has a finite number of simplices for all $s \leq t$ (this condition holds for the \v{C}ech complex and the Vietoris--Rips complex), then there is a sequence $\{r_i\}$ such that $\K_{r_i} = \K_r$ for all $r_i \leq r < r_{i+1}$. The direct sum $\bigoplus_i H(\K_{r_i}, \mathbb{F})$ has the structure of a graded module over the graded ring $\mathbb{F}[x]$. The action of $x$ on a homogenous element $\gamma \in H(\K_{r_i}, \mathbb{F})$ is $x\gamma = \iota_{r_i}^{r_{i+1}}(\gamma)$.
\begin{theorem}[Fundamental Theorem of Persistent Homology \cite{fundthm}]
The graded $\mathbb{F}[x]$-module $\bigoplus_i H(\K_{r_i}, \mathbb{F})$ is isomorphic to
\begin{equation}\label{eq:phmodule}
    \Big(\bigoplus_i \Sigma^{a_i} \mathbb{F}[x] \Big) \oplus \Big(\bigoplus_j \Sigma^{b_j} \mathbb{F}[x]/(x^{c_j})\Big)
\end{equation}
for some integers $\{a_i\}$, $\{b_j\}$, $\{c_j\}$, where $\Sigma^m \mathbb{F}[x]$ denotes an $m$-shift upward in grading for any integer $m$.
\end{theorem}
An $\Sigma^{a_i} \mathbb{F}[x]$ summand corresponds to a homology class that is born at filtration level $r_{a_i}$ and never dies. An $\Sigma^{b_j} \mathbb{F}[x]/(x^{c_j})$ summand corresponds to a homology class that is born at filtration level $r_{b_j}$ and dies at filtration level $r_{c_j} + r_{b_j}$. The information in a persistence module can be summarized by a \emph{persistence diagram}, which is a multiset of points in the extended plane $\overline{\mathbb{R}}^2$. Given a decomposition in the form of Equation \ref{eq:phmodule}, the persistence diagram includes the points $(r_{a_i}, \infty)$ for all $i$, the points $(r_{b_j}, r_{c_j})$ for all $j$, and all points on the diagonal. The points on the diagonal are included for technical reasons; one can think of them as homology classes that die instantaneously. We denote the persistence diagram of a persistence module $\mathbb{V}$ by $\dgm(\mathbb{V})$. The \emph{bottleneck distance} between two diagrams is defined to be
\begin{equation*}
    W_{\infty}(\dgm(\mathbb{V}), \dgm(\mathbb{U})) :=  \inf_{\eta} \sup_{x \in \dgm(\mathbb{V})} \norm{x - \eta(x)}_{\infty}\,,
\end{equation*}
where the infimum is taken over all bijections $\eta: \dgm(\mathbb{V}) \to \dgm(\mathbb{U})$.


\subsection{Riemannian Geometry}\label{sec:Rgeom}
We briefly review the necessary background from Riemannian geometry. For further reading, we recommend a textbook such as \cite{petersen}. A \emph{Riemannian manifold} $(M, g)$ is a smooth manifold $M$ with a Riemannian metric $g$ that defines a smoothly-varying inner product on each tangent space $T_xM$. More precisely, $g$ is a 2-tensor field on $M$; to each $x \in M$, the Riemannian metric $g$ assigns a bilinear map $g_x$ on the tangent space $T_xM$. A Riemannian metric allows one to define the length of a vector $v \in T_xM$ to be $\norm{v} := g_x(v, v)^{1/2}$. The length of a continuously differentiable path $\gamma: [a, b] \to M$ is defined to be $L(\gamma) := \int_a^b \norm{\gamma'(t)}\mathrm{d}t$.

A Riemannian manifold is a metric space. The distance between two points $x$, $y$ in the same connected component of $M$ is
\begin{equation*}
    d_{M, g}(x, y) := \inf \{L(\gamma) \mid \gamma:[a, b] \to M \text{ is a $C^1$ path such that } \gamma(a) = x \text{ and }\gamma(b) = y\}.
\end{equation*}
If $(M, g)$ is \emph{complete}, then the infimum is achieved by a \emph{geodesic}, a curve that locally minimizes length. If $x$ and $y$ are in different connected components, then their distance is infinite.

To see that all manifolds can be given a Riemannian metric, recall that all manifolds can be embedded into Euclidean space. Let $\iota: M \xhookrightarrow{} \mathbb{R}^m$ be an embedding. The canonical Euclidean metric $\bar{g}$ pulls back to a Riemmanian metric $\iota^* \bar{g}$ on $M$. We call $\iota^*\bar{g}$ the \emph{Euclidean-induced Riemannian metric}. On each tangent space $T_xM$, the metric $\iota^*\bar{g}_x$ is the restriction of $\bar{g}_x$ to $T_x M$. A Riemannian metric induces a \emph{volume form} $dV$, the unique $n$-form on $M$ that equals $1$ on all positively oriented orthonormal bases. In local coordinates, the expression for the volume form is
\begin{equation*}
    dV = \sqrt{\vert g \vert} dx^1 \wedge \cdots \wedge dx^n.
\end{equation*}

With a volume form and a smooth probability density function $f: M \to (0, \infty)$, one can define a probability measure on the manifold. A good reference for probability and statistics on Riemannian manifolds is \cite{Rman_stats}. The volume form induces a \emph{Riemannian measure} $\mu$ on $M$. The measure of a Borel set $A \subseteq M$ is $\mu(A) = \int_A dV$, and the volume of $M$ is $\mu(M)$. The probability measure is defined to be
\begin{equation*}
    \p[A] = \int_{A} f dV
\end{equation*}
for Borel sets $A \subseteq M$.

Two Riemannian metrics $g$, $\g$ on $M$ are \emph{conformally equivalent} if there is a positive $C^{\infty}$ function $\lambda: M \to \mathbb{R}$ such that
\begin{equation*}
    \g = \lambda g.
\end{equation*}
A \emph{conformal transformation} is a diffeomorphism $F: (M, g) \to (M', g')$ such that $g'$ pulls back to $\lambda g$ (i.e., $F^*g' = \lambda g$) for some positive $C^{\infty}$ function $\lambda: M \to \mathbb{R}$. Conformal transformations preserve angles; one can think of a conformal transformation as a transformation that ``locally scales'' the manifold. For example, if $M$ is a submanifold of $\mathbb{R}^m$ and has the Euclidean-induced Riemannian metric, then any global scaling is a conformal transformation.

A special type of conformal transformation is an isometry. An \emph{isometry} of Riemannian manifolds is a diffeomorphism $F: (M, g) \to (M', g')$ such that $g'$ pulls back to $g$ (i.e., $F^*g' = g$). An isometry of Riemannian manifolds is an isometry of metric spaces in the usual sense (i.e., $d_{M, g}(x, y) = d_{M', g'}(F(x), F(y))$). 


\section{Our Family of Density-Scaled Filtered Complexes}\label{sec:defs}


\subsection{Our Density-Scaled Riemannian Manifold}\label{sec:def_Mscaled}
Let $(M, g)$ be an $\Mdim$-dimensional Riemannian manifold from which we sample $N$ points according to a smooth probability density function $f: M \to (0, \infty)$. We begin by defining a conformally-equivalent Riemannian metric $\g$ such that the points are uniformly distributed in $(M, \g)$.
\begin{definition}
The \emph{density-scaled Riemannian metric} is
\begin{equation}\label{eq:scaled_g}
    \g = \sqrt[n]{\alpha(N)^2f^2}g\,,
\end{equation}
where $\alpha(N)$ is a strictly positive function that satisfies
\begin{equation}\label{eq:alpha_conditions}
    \alpha(N) \to \infty\,, \qquad
    \frac{\alpha(N)\Lambda^*}{N} \to 0
\end{equation}
as $N \to \infty$, and where $\Lambda^*$ is the threshold filling factor defined in Equation \ref{eq:filling_thresh} in Section \ref{sec:convergence}. 
\end{definition}
In this paper, we set
\begin{equation*}
    \alpha(N) := \begin{cases}
        \frac{N}{\log N (\log N + (n-1)\log \log N)} \,, & N > 1 \\
        1 \,, & N = 1\,,
    \end{cases}
\end{equation*}
which satisfies Equation \ref{eq:alpha_conditions}. However, the convergence properties of Sections \ref{sec:convergence} hold for any choice of $\alpha(N)$ that satisfies the conditions of Equation \ref{eq:alpha_conditions}, and the invariance and stability results in Sections \ref{sec:invariance} and \ref{sec:stability} hold for \emph{any} choice of strictly positive function $\alpha(N)$.

The uniform probability measure on $(M, \g)$ is $\p[A] = \int_A \frac{1}{\tilde{\mu}(M)}d\widetilde{V}$ for all Borel sets $A \subseteq M$, where $d\widetilde{V}$ is the volume form on $(M, \g)$ and $\tilde{\mu}(M)$ is the volume of $(M, \g)$. Using local coordinates, we see that $d\widetilde{V}$ satisfies
\begin{equation*}
    d\widetilde{V} = \sqrt{\vert \g \vert}dx^1 \wedge \cdots \wedge dx^n = \alpha(N)f \sqrt{\vert g \vert}dx^1 \wedge \cdots \wedge dx^n = \alpha(N)f dV.
\end{equation*}
Therefore $\frac{1}{\tilde{\mu}(M)}d\widetilde{V} = fdV$ because
\begin{equation*}
    \tilde{\mu}(M) = \alpha(N)\int_M fdV = \alpha(N)\,.
\end{equation*}
This means that sampling points from $(M, g)$ with probability density function $f$ is equivalent to sampling points uniformly at random from $(M, \g)$.


\subsection{Our Definition of a Density-Scaled Filtered Complex}
\begin{definition}
Let $(M, g)$ be a Riemannian manifold, and let $X = \{x_1, \ldots, x_N\}$ be a point cloud that consists of $N$ points sampled from a smooth probability density function $f: M \to(0, \infty)$. The \emph{density-scaled C\v{e}ch complex} is the filtered complex
\begin{equation*}
    \dc(M, g, f, X) := \C(M, d_{M, \g}, X),
\end{equation*}
where $d_{M, \g}$ is the Riemannian distance function in $(M, \g)$ and $\g$ is defined as in Equation \ref{eq:scaled_g}. Equivalently, the set of simplices in $\dc(M, g, f, X)_r$ at filtration level $r$ is 
\begin{equation*}
    \Big\{x_J \mid \bigcap_{j \in J} B(x_j, r) \neq \emptyset \text{ and } J \subseteq \{1, \ldots, N \}\Big\},
\end{equation*}
where $B(x, r) := \{y \in M \mid d_{M, \g}(x, y) \leq r\}$.
\end{definition}

\begin{definition}
Let $(M, g)$ be a Riemannian manifold, and let $X = \{x_1, \ldots, x_N\}$ be a point cloud that consists of $N$ points sampled from a smooth probability density function $f: M \to(0, \infty)$. The \emph{density-scaled Vietoris--Rips complex} is the filtered complex 
\begin{equation*}
    \textnormal{DVR}(M, g, f, X) := \textnormal{VR}(M, d_{M, \g}, X)\,,
\end{equation*}
where $d_{M, \g}$ is the Riemannian distance function in $(M, \g)$ and $\g$ is defined as in Equation \ref{eq:scaled_g}. Equivalently, the set of simplices in $\textnormal{DVR}(M, g, f, X)_r$ at filtration level $r > 0$ is
\begin{equation*}
    \Big\{x_J \mid d_{M, \g}(x_i, x_j) \leq 2r \text{ for all } i, j \in J \text{ and } J \subseteq \{1, \ldots, N\}\Big\}.
\end{equation*}
\end{definition}

More generally, one can define a density-scaled version of any distance-based filtered complex by applying the filtered complex to the point cloud in the metric space $(M, d_{M, \g})$, where $d_{M, \g}$ is the Riemannian distance function in the density-scaled manifold $(M, \g)$.
\begin{definition}
Let $(M, g)$ be a Riemannian manifold, and let $X = \{x_1, \ldots, x_N\}$ be a point cloud that consists of $N$ points sampled from a smooth probability density function $f: M \to(0, \infty)$. If $\Sigma(M, d, X)$ is a distance-based filtered complex, where $(M, d)$ denotes a metric space, then the \emph{density-scaled filtered complex} is
\begin{equation*}
D\Sigma(M, g, f, X) := \Sigma(M,  d_{M, \g}, X)\,,
\end{equation*}
where $d_{M, \g}$ is the Riemannian distance function in $(M, \g)$ and $\g$ is defined as in Equation \ref{eq:scaled_g}.
\end{definition}


\section{Convergence Properties of the Density-Scaled \v{C}ech Complex}\label{sec:convergence}
In Theorem \ref{thm:convergence} below, we show that the density-scaled \v{C}ech complex is homotopy-equivalent to $M$ for an interval of filtration values that grows arbitrarily large in probability as $N \to \infty$. We begin by reviewing the relevant concepts. The \emph{convexity radius} of a Riemannian manifold $(M, g)$ is 
\begin{equation*}
    r^{\textnormal{convex}} := \sup\{r \mid B(x, s) \text{ is geodesically convex for all } x \in M \text{ and all } 0 \leq s < r\}\,,
\end{equation*}
where $B(x, s) := \{y \in M \mid d_{M, g}(x, y) \leq r\}$ and where $d_{M, g}$ is the Riemannian distance function in $(M, g)$. If $s < r^{\textnormal{convex}}$, the ball $B(x, s)$ is geodesically convex (hence contractible). Furthermore, the intersection of geodesically convex balls is geodesically convex (hence contractible or empty). Let $\g_N$ denote the density-scaled Riemannian metric when there are $N$ points, and let $r^{\textnormal{convex}}_N$ denote the convexity radius of $(M, \g_N)$. The \emph{coverage radius} of a point cloud $X$ in a Riemannian manifold $(M, g)$ is
\begin{equation*}
    r^{\text{cover}} := \inf\Big\{r \mid M \subseteq \bigcup_{x \in X} B(x, r) \Big\}.
\end{equation*}
Let $r^{\text{cover}}_N$ denote the coverage radius of a point cloud $X$ in $(M, \g_N)$.

\begin{theorem}\label{thm:convergence}
Let $(M, g)$ be a Riemannian manifold, and let $X$ be a point cloud that consists of $N$ points sampled from a smooth probability density function $f: M \to (0, \infty)$. If $r^{\textnormal{cover}}_N < r < r^{\textnormal{convex}}_N$, then $\dc(M, g, f, X)_r$ is homotopy-equivalent to $M$. If $M$ is compact, then $r^{\textnormal{convex}}_N \to \infty$ as $N \to \infty$. If $M$ is compact and connected, then $r^{\textnormal{cover}}_N \to 0$ in probability as $N \to \infty$.
\end{theorem}
\begin{proof}
If $r < r^{\text{convex}}_N$, then for all $J \subseteq \{1, \ldots, N\}$, the intersection $\bigcap_{j \in J} B(x_j, r)$ is convex, so it is either contractible or empty. If $r > r^{\text{cover}}_N$, then $\bigcup_i B(x_i, r) = M$. By the Nerve Theorem, $\dc(M, g, f, X)_r$ is homotopy-equivalent to $M$. The second statement of the theorem follows from Lemma \ref{lem:conv_inf} below, and the third statement of the theorem follows from Lemma \ref{lem:cover_zero} below.
\end{proof}

\begin{lemma}\label{lem:conv_inf}
If $M$ is compact, then $r^{\textnormal{convex}}_N \to \infty$ as $N \to \infty$.
\end{lemma}
\begin{proof}
The convexity radius of a compact manifold is positive (see, e.g., Proposition 20 in \cite{convexity_radius}). Therefore, $r^{\text{convex}}_1 > 0$, so $r^{\text{convex}}_N = \sqrt[n]{\alpha(N)} r^{\text{convex}}_1 \to \infty$ because $\alpha(N) \to \infty$.
\end{proof}

Now we turn to the coverage radius. The behavior of the coverage radius is controlled by the \emph{filling factor}. On an $n$-dimensional Riemannian manifold $M$ from which $N$ balls of radius $r$ are chosen uniformly at random, the filling factor is
\begin{equation}
    \Lambda := N v_nr^n/\mu(M),
\end{equation}
where $v_n$ is the volume of a Euclidean unit $n$-ball. For small $r$, the filling factor approximates the number of points inside a ball of radius $r$. Let $N_r$ be the number of balls of radius $r$ required to cover $M$, assuming the balls are chosen uniformly at random. Let $\eta = v_n r^n$ be the volume of a Euclidean $n$-ball of radius $r$. Define
\begin{equation*}
    X_r := \eta N_r - \log (\eta^{-1}) - n \log \log(\eta^{-1}).
\end{equation*}
\begin{theorem}[Theorem 1.1 in \cite{flatto_newman}]\label{thm:flatto}Let $M$ be a compact, connected Riemannian manifold with unit volume. There are constants $r_1 > 0$ and $C > 0$, which do not depend on $M$, such that if $r \leq r_1$, then
\begin{align*}
    \p[X_r > x] &\leq Ce^{-x/8} \,,  \qquad{x \geq 0} \\
    \p[X_r < x] &\leq Ce^x \,, \qquad{x \leq 0}.
\end{align*}
\end{theorem}
\begin{corollary}\label{cor:cover_prob}
Let $(M, g)$ be a compact, connected Riemannian manifold. Suppose $X$ is a point cloud that consists of $N$ points sampled uniformly at random from $M$. Suppose $w(N)$ is a sequence such that $w(N) \to \infty$ and $w(N)/\log N \to 0$.
\begin{enumerate}
\item If $r$ is such that $\Lambda = \log N + (n-1)\log \log N + w(N)$, then $\p[r^{\textnormal{cover}} > r] \to 0$ as $N \to \infty$.
\item If $r$ is such that $\Lambda = \log N + (n-1) \log \log N - w(N)$, then $\p[r^{\textnormal{cover}} < r] \to 0$ as $N \to \infty$.
\end{enumerate}
\end{corollary}
\begin{proof}
\underline{Case 1: $\mu(M) = 1$} 

In this case, the structure of our proof is similar to that of Corollary B.2 in \cite{vanishing_homology}. First, we observe that the radius of the balls can be expressed by
\begin{equation*}
    r = \sqrt[n]{\frac{\mu(M)\Lambda}{Nv_n}}.
\end{equation*}
If $N$ is sufficiently large and $\Lambda = \log N + (n-1)\log\log N \pm w(N)$, then $r < r_1$. Let $x = \eta N - \log(\eta^{-1}) - n \log \log (\eta^{-1})$. 
\begin{enumerate}
\item If $\Lambda = \log N + (n-1)\log \log N + w(N)$, then
\begin{equation}\label{eq:coverage_upper}
    \p[r^{\text{cover}} > r] = \p[N_r > N] = \p[X_r > x] \leq Ce^{-x/8}
\end{equation}
by Theorem \ref{thm:flatto}. Because $\eta = N^{-1}\Lambda$, we have
\begin{equation*}
    x = \Lambda - \log N + \log \Lambda - n \log(\log N - \log \Lambda).
\end{equation*}
We expand the first term as $\Lambda = \log N + (n-1)\log \log N + w(N)$ to get
\begin{equation*}
    x = [\log \Lambda - \log \log N] + n[\log \log N - \log (\log N - \log \Lambda)] + w(N).
\end{equation*}
Because $w(N)/\log N \to 0$, we have
\begin{align*}
    \lim_{N \to \infty} [\log \Lambda - \log \log N] &= \log\Big(\lim_{N \to \infty} \frac{\Lambda}{\log N}\Big) = 0\,, \\
    \lim_{N \to \infty}[\log \log N - \log (\log N - \log \Lambda)] &= - \log\Big(1 - \lim_{N \to \infty} \frac{\log \Lambda}{\log N}\Big) = 0\,.
\end{align*}
Therefore, $x \to \infty$ as $N \to \infty$. By Equation \ref{eq:coverage_upper}, $\p[r^{\text{cover}} > r] \to 0$.

\item If $\Lambda = \log N + (n-1)\log \log N - w(N)$, then
\begin{equation}\label{eq:coverage_lower}
    \p[r^{\text{cover}} < r] = \p[N_r < N] = \p[X_r < x] \leq Ce^x
\end{equation}
by Theorem \ref{thm:flatto}. Similarly to above, we have
\begin{equation*}
    x = [\log \Lambda - \log \log N] + n[\log \log N - \log (\log N - \log \Lambda)] - w(N)\,,
\end{equation*}
so $x \to -\infty$ as $N \to \infty$. By Equation \ref{eq:coverage_lower}, $\p[r^{\text{cover}} < r] \to 0$.
\end{enumerate}

\underline{Case 2: $\mu(M) \neq 1$}

Let $(M, \overline{g})$ be the Riemannian manifold that is normalized to have unit volume. Let $\overline{\Lambda}$ denote the filling factor for $(M, \overline{g})$ and let $\overline{r^{\text{cover}}}$ be the coverage radius for the point cloud $X$ in $(M, \overline{g})$. In $(M, \overline{g})$, the Riemannian distance function is $d_{M, \overline{g}} = \sqrt[-n]{\mu(M)}d_{M, g}$. Therefore, for any $r$, we have
\begin{align*}
    \p[r^{\text{cover}} > r] &= \p\Big[\overline{r^{\text{cover}}} > r/\sqrt[n]{\mu(M)}\Big]\,, \\
    \p[r^{\text{cover}} < r] &= \p\Big[\overline{r^{\text{cover}}} < r/\sqrt[n]{\mu(M)}\Big].
\end{align*}
When the radius of the balls in $(M, \overline{g})$ is $r/\sqrt[n]{\mu(M)}$, the filling factor in $(M, \overline{g})$ is 
\begin{equation*}
    \overline{\Lambda} = N v_n\Big(r/\sqrt[n]{\mu(M)}\Big)^n = \Lambda.
\end{equation*}
Applying Case 1 to $(M, \overline{g})$ completes the proof.
\end{proof}
This shows that on a compact, connected Riemannian manifold $M$ from which points are sampled uniformly at random, there is a threshold filling factor 
\begin{equation}\label{eq:filling_thresh}
    \Lambda^* := \log N + (n-1)\log \log N\,,
\end{equation}
above which the balls are likely to cover $M$ and below which the balls are unlikely to cover $M$. There is a corresponding threshold radius $r^* := \sqrt[n]{\frac{\mu(M) \Lambda^*}{Nv_n}}$. The threshold radius on $(M, \g_N)$ is
\begin{equation}\label{eq:threshold}
    r^*_N :=  \sqrt[n]{\frac{\alpha(N)\Lambda^*}{Nv_n}}.
\end{equation}
By Equation \ref{eq:alpha_conditions}, we have that $r^*_N \to 0$ as $N \to \infty$.
\begin{lemma}\label{lem:cover_zero}
Let $M$ be a compact, connected Riemannian manifold, and let $f: M \to (0, \infty)$ be a smooth probability density function from which $N$ points are sampled. Then $r^{\textnormal{cover}}_N \to 0$ in probability as $N \to \infty$. Moreover, $r^{\textnormal{cover}}_N/r^*_N \to 1$ in probability.
\end{lemma}
\begin{proof}
Let $w(N)$ be a sequence such that $w(N) \to \infty$ and $w(N)/\log N \to 0$. Define the sequence of filling factors
\begin{equation*}
    \Lambda_N^{\pm} = \log N + (n-1)\log \log N \pm w(N)\,,
\end{equation*}
and define $r_N^{\pm}$ to be
\begin{equation*}
    r_N^{\pm} :=  \sqrt[n]{\frac{\alpha(N)\Lambda_N^{\pm}}{Nv_n}}\,,
\end{equation*}
which is the radius that corresponds to a filling factor of $\Lambda_N^{\pm}$ on $(M, \g_N)$. Note that $\frac{r_N^{\pm}}{r^*_N} \to 1$, where $r^*_N$ is defined as in Equation \ref{eq:threshold}. Because $r^*_N \to 0$, it must be true that $r_N^{\pm} \to 0$.

Let $\epsilon > 0$. For sufficiently large $N$, we have $r_N^+ < \epsilon$, so $\p[r^{\text{cover}}_N > \epsilon] < \p[r^{\text{cover}}_N > r_N^+]$. Applying Corollary \ref{cor:cover_prob} proves the first statement of the lemma. For sufficiently large $N$, we have 
\begin{equation*}
    1 - \epsilon \leq \frac{r_N^-}{r_N^*}  < 1 < \frac{r_N^+}{r_N^*} \leq 1 + \epsilon\,,
\end{equation*}
so
\begin{equation*}
\p\Big[\Big|\frac{r^{\text{cover}}_N}{r^*_N} - 1\Big| < \epsilon\Big] > \p\Big[\frac{r_N^-}{r^*_N} < \frac{r^{\text{cover}}_N}{r^*_N} < \frac{r_N^+}{r{*, N}}\Big] = \p[r_N^- < r^{\text{cover}}_N < r_N^+].
\end{equation*}
Applying Corollary \ref{cor:cover_prob} completes the proof.
\end{proof}


\section{Conformal Invariance}\label{sec:invariance}

Let $(M_1, g_1)$ and $(M_2, g_2)$ be Riemannian manifolds, and let $F: M_2 \to M_1$ be a diffeomorphism. If $f_1: M_1 \to (0, \infty)$ is a smooth probability density function, then we can pull back $f_1$ to a probability density function $f_2: M_2 \to (0, \infty)$ as follows.
\begin{definition}[Pullback of a Probability Density Function]\label{def:pdf_pullback}
The \emph{pullback of $f_1$ under $F$} is the function $f_2 : M_2 \to (0, \infty)$ such that $f_2dV_2 = F^*(f_1dV_1)$. The probability density function $f_2$ exists because the space of $n$-forms on an $n$-dimensional manifold is spanned by $dV_2$.
\end{definition}
The pullback of a probability density function is defined such that sampling a point cloud $Y$ from  $f_2$ is equivalent to sampling a point cloud $X$ from $f_1$ and setting $Y = F^{-1}(X)$.
\begin{proposition}\label{prop:densitypullback}
Suppose $x$ is sampled from $f_1 : M_1 \to (0, \infty)$ and let $y_1 = F^{-1}(x)$. Suppose $y_2$ is sampled from $f_2 : M_2 \to (0, \infty)$, where $f_2$ is the pullback of $f_1$ defined by Definition \ref{def:pdf_pullback}. Then $y_1$ and $y_2$ are identically distributed.
\end{proposition}
\begin{proof}
If $A \subseteq M_2$ is a Borel set, then
\begin{equation*}
    \p[y_1 \in A] = \p[x \in F(A)] = \int_{F(A)} f_1dV_1 = \int_A f_2 dV_2 = \p[y_2 \in A]\,.
\end{equation*}
\end{proof}
Prop \ref{prop:densitypullback} justifies a comparison of $D\Sigma(M_1, g_1, f_1, \cdot)$ to $D\Sigma(M_2, g_2, f_2, F^{-1}(\cdot))$. Below, we define what we mean by an \emph{isomorphism} of two filtered complexes and what we mean by \emph{invariance} of a filtered complex.
\begin{definition}[Isomorphism of Filtered Complexes]\label{def:isomorphism}
Let $\K^1 = \{\K^1_r\}_{r \in \mathbb{R}}$ and $\K^2 = \{K^2_r\}_{r \in \mathbb{R}}$ be filtered complexes, and let $V_r^i$, $S_r^i$ be the sets of vertices and simplices, respectively, of $K^i_r$. Let $V^i = \bigcup_r V^i_r$ be the set of all vertices of $\K^i$. We say that $\K^1$ and $\K^2$ are \emph{isomorphic} if there is a bijective map $\phi: V^1 \to V^2$ such that $\phi$ induces bijections $V_r^1 \to V_r^2$ and $S_r^1 \to S_r^2$ for all $r$.
\end{definition}
\begin{definition}[Invariance]\label{def:pointwise_invar}
Let $(M_1, g_1)$ and $(M_2, g_2)$ be Riemannian manifolds, and let $F: M_2 \to M_1$ be a diffeomorphism. A density-scaled complex $D\Sigma$ is \textit{invariant under $F$} if $D\Sigma(M_1, g_1, f_1, X)$ is isomorphic to $D\Sigma(M_2, g_2, f_2, F^{-1}(X))$ for all smooth probability density functions $f_1: M_1 \to (0, \infty)$ and point clouds $X$ sampled from $f_1$, where $f_2: M_2 \to (0, \infty)$ is the pullback of $f_1$ defined by Definition \ref{def:pdf_pullback}.
\end{definition}

We restrict ourselves to a suitable class of distance-based filtered complexes $\Sigma(M, d, X)$ that are \emph{invariant under global isometry}. This class includes the \v{C}ech complex, the Vietoris--Rips complex, and many other standard distance-based filtered complexes.

\begin{definition}[Invariance Under Global Isometry] Let $(M_1, d_1)$ and $(M_2, d_2)$ be metric spaces. A distance-based filtered complex $\Sigma(M, d, X)$ is \emph{invariant under global isometry} if $\Sigma(M_1, d_1, X)$ is isomorphic to $\Sigma(M_2, d_2, F^{-1}(X))$ for all global isometries $F: M_2 \to M_1$ and all point clouds $X$ in $M_1$.
\end{definition}

Theorem \ref{thm:local_invar} below shows in particular that the density-scaled \v{C}ech complex $\dc$ and the density-scaled Vietoris--Rips complex DVR are invariant under all conformal transformations. As a corollary, this implies that they are invariant under global scaling (Corollary \ref{cor:global_invar}). Additionally, they are invariant under diffeomorphisms of $1$-dimensional manifolds (Corollary \ref{cor:1dim_invar}).

\begin{theorem}\label{thm:local_invar} 
Suppose that $\Sigma(M, d, X)$ is a distance-based filtered complex that is invariant under global isometry, and let $D\Sigma$ be the density-scaled filtered complex. Then $D\Sigma$ is invariant under all conformal transformations.
\end{theorem}
\begin{proof}
Suppose that $(M_1, g_1)$ and $(M_2, g_2)$ are Riemannian manifolds with Riemannian distance functions $d_{M_1, g_1}$ and $d_{M_2, g_2}$, respectively. Let $f_1: M_1 \to (0, \infty)$ be a probability density function on $M_1$. Suppose $F: M_2 \to M_1$ is a conformal transformation, and let $f_2$ be the pullback of $f_1$ defined by Definition \ref{def:pdf_pullback}. Let $\g_1$, $\g_2$ be the density-scaled Riemannian metrics, and suppose that $X$ is a point cloud that consists of $N$ points sampled from $f_1$. By Lemma \ref{lem:F_induces_iso} in Appendix \ref{sec:lemmas}, $\Sigma(M_1, d_{M_1, \g_1}, X)$ is isomorphic to $\Sigma( M_2, d_{M_2, \g_2}, F^{-1}(X))$ because $\Sigma$ is invariant under global isometry. Thus $D\Sigma(M_1, g_1, f_1, X)$ is isomorphic to $D\Sigma(M_2, g_2, f_2, F^{-1}(X))$. 
\end{proof}
\begin{corollary}[Density-Scaled Complexes are Invariant Under Global Scaling]\label{cor:global_invar}
Let $M$ be a submanifold of $\mathbb{R}^m$ with the Euclidean-induced Riemannian metric $g_M$. Suppose $L:\mathbb{R}^m \to \mathbb{R}^m$ is the linear transformation $L(\bm{x}) = \lambda\bm{x}$ for some $\lambda \in (0, \infty)$. If $\Sigma(M, d, X)$ is a distance-based filtered complex that is invariant under global isometry, then the density-scaled complex $D\Sigma$ is invariant under $L^{-1}$.
\end{corollary}
\begin{proof}
Let $L(M)$ be the image of $M$ under $L$, and let $g_{L(M)}$ be the Euclidean-induced Riemannian metric on $L(M)$. The map $L^{-1}$ is a conformal transformation because $(L^{-1})^*g_M = \lambda^{-2} g_{L(M)}$.
\end{proof}
\begin{corollary}\label{cor:1dim_invar}
Suppose that $\Sigma(M, d, X)$ is a distance-based filtered complex that is invariant under global isometry, and let $D\Sigma$ be the density-scaled complex. Then $D\Sigma$ is invariant under all diffeomorphisms of $1$-dimensional manifolds.
\end{corollary}
\begin{proof}
Let $F: M_2 \to M_1$ be a diffeomorphism between $1$-dimensional Riemannian manifolds $(M_1, g_1)$ and $(M_2, g_2)$. Because each tangent space is $1$-dimensional, we must have that $F^*g_1 = \lambda g_2$ for some positive smooth function $\lambda: M_2 \to \mathbb{R}$.
\end{proof}


\section{Implementation}\label{sec:implement}
For an $n$-dimensional Riemannian manifold $M$ that is a submanifold of $\mathbb{R}^m$ and has the Euclidean-induced Riemannian metric $g$, we implement a filtered complex $\dvr(n, k, X)$ that approximates the density-scaled Vietoris--Rips complex $\textnormal{DVR}(M, g, f, X)$. The implementation requires a choice of parameter $k$ (see Section \ref{sec:geo_est}), knowledge of the dimension $n$ of the manifold, and knowledge of the pairwise Euclidean distances between the points of $X$. The dimension $n$ can be estimated using one of the methods mentioned previously \cite{local_PCA, intrinsic_dim, conical_dimestimate, ball_dimestimate, doubling_dimestimate}, and we describe a heuristic method for choosing $k$ at the end of Section \ref{sec:geo_est}.


\subsection{Estimation of $f$}\label{sec:kde}
We estimate the probability density function $f$ by using kernel-density estimation. As described in \cite{kde_submanifold}, we use an $n$-dimensional kernel, where $n = \dim(M)$.

\begin{theorem}[Theorem 2.1 in \cite{kde_submanifold}]\label{thm:kde}
Suppose $M$ is an $n$-dimensional submanifold of $\mathbb{R}^m$. Let $K:\mathbb{R} \to [0, \infty)$ be a kernel function such that
\begin{enumerate}
    \item $K$ is symmetric: $K(-x) = K(x)$,
    \item $K$ is normalized: $\int_{\norm{\bm{z}} \leq 1} K(\norm{\bm{z}}) \mathrm{d}^n\bm{z} = 1$,
    \item $K(x) = 0$ for $x \not\in (-1, 1)$, and
    \item $K$ is differentiable in $(-1, 1)$, with bounded derivative.
\end{enumerate}
Let $\{h_N\}_{N \in \mathbb{N}}$ be a sequence of bandwidth parameters in $(0, \infty)$. Given a point cloud $X$ that consists of $N$ points sampled from a probability density function $f: M \to (0, \infty)$, the estimator of $f$ is defined to be
\begin{equation}\label{eq:fest}
    \hat{f}_N(y) := \frac{1}{N}\sum_{x \in X} \frac{1}{h_N^n}K\Bigg(\frac{\norm{y - x}}{h_N}\Bigg)
\end{equation}
for all $y \in \mathbb{R}^m$, where $\norm{y- x}$ denotes the Euclidean distance in $\mathbb{R}^m$. If $f$ is twice differentiable in a neighborhood of $y$ and $h_N \propto N^{\frac{-1}{n+4}}$, then 
\begin{equation*}
    \e[(\hat{f}_N(y) - f(y))^2] = O\Bigg(\frac{1}{N^{\frac{4}{n+4}}}\Bigg) \qquad \text{as } N \to \infty\,.
\end{equation*}
\end{theorem}
As $h_N \to 0$, the condition that $K$ has compact support ensures that we are only averaging around a small neighborhood of $y$. This is important on a manifold because $\norm{y-x}$ is only guaranteed to be a good approximation to the Riemannian distance $d_{M, g}(x, y)$ when $y$ is close to $x$ (see, e.g., Lemma 3.1 in \cite{kde_submanifold}).

In our implementation, we set the bandwidth to $h_N = N^{\frac{-1}{n+4}}$ (Scott's rule \cite{scotts_rule}), which satisfies the conditions of Theorem \ref{thm:kde}. Kernels that satisfy the conditions of Theorem \ref{thm:kde} include the Epanechnikov kernel, the biweight kernel, and the triweight kernel \cite{kernels}. We set the default kernel $K$ to be the biweight kernel because it has the most consistent performance in our experiments of Section \ref{sec:compare}. (In the example of Section \ref{sec:circles}, we explore the effects of the choice of kernel.) In dimension $n = 1$, the biweight kernel is
\begin{equation}\label{eq:biweight_dim1}
    K(x) := \begin{cases}
        \frac{15}{16}(1-x^2)^2 \,, & x \in (-1, 1) \\
        0 \,, & \text{otherwise.}
    \end{cases}
\end{equation}
In higher dimensions $n$, the biweight kernel of Equation \ref{eq:biweight_dim1} must be normalized differently. In general, if $K(x)$ is a kernel function in dimension $n = 1$, then the radial kernel function in dimension $n$ is
\begin{equation*}
    K_n(x) = \frac{K(x)}{s_{n-1}\int_0^1 K(r) r^{n-1}\mathrm{d}r},
\end{equation*}
where $s_{n-1}$ is the surface area of the $(n-1)$-sphere. For example, the biweight kernel in dimension $n$ is
\begin{equation*}
    K_n(x) := \begin{cases}
    \Big(s_{n-1}(\frac{1}{n} - \frac{2}{n+2} + \frac{1}{n+4})\Big)^{-1}(1-x^2)^2 \,, & x \in (-1, 1) \\
    0 \,, & \text{otherwise.}
    \end{cases}
\end{equation*}


\subsection{Estimation of Riemannian Distances in the Density-Scaled Manifold}\label{sec:geo_est}
Let $d_{M, \g}$ denote the Riemannian distance function in $(M, \g)$. In a similar manner as to how Riemannian distances are estimated in Isomap \cite{isomap} and C-Isomap \cite{Cisomap}, we estimate $d_{M, \g}$ as follows.
\begin{enumerate}
    \item Construct the $k$-nearest neighbor graph $G_{kNN}(X)$ for some choice of parameter $k$. (A heuristic method for choosing $k$ is discussed at the end of this subsection.) Vertices $x_i, x_j$ are connected by an edge if either $x_j$ is one of the $k$-nearest neighbors of $x_i$ (as measured by Euclidean distance) or if $x_i$ is one of the $k$-nearest neighbors of $x_j$ (as measured by Euclidean distance).
    
    \item We set the weight of an edge $(x_i, x_j) \in G_{kNN}(X)$ to
    \begin{equation*}
        w(x_i, x_j) := \sqrt[n]{\alpha(N)\max\{\hat{f}_N(x_i), \hat{f}_N(x_j)\}} \norm{x_i - x_j}\,,
    \end{equation*}
    where $\hat{f}$ is defined as in Equation \ref{eq:fest}.
    
    \item For any $x_i, x_j \in X$, our estimate $\widehat{d_{M, \g}}(x_i, x_j)$ of $d_{M, \g}(x_i, x_j)$ is the length of the shortest weighted path in $G_{kNN}(X)$ from $x_i$ to $x_j$, if such a path exists. We set $\widehat{d_{M, \g}}(x_i, x_j) = \infty$ if $x_i$ and $x_j$ are not in the same connected component of $G_{kNN}(X)$.
\end{enumerate}

In step 1, we connect each point to its $k$-nearest neighbors. When $N$ is large, the $k$-nearest neighbors to a point $x$ are likely to be within a small neighborhood of $x$. If $y$ is near $x$, then $\sqrt[n]{\alpha(N)f(x)}d_{M, g}(x, y)$ is a good approximation to $d_{M, \g}(x,y)$. That is,
\begin{equation}\label{eq:conformaldist_approx}
    \frac{d_{M, \g}(x, y)}{\sqrt[n]{\alpha(N)f(x)}d_{M, g}(x,y)} \to 1 \qquad{\text{ as } x \to y}\,.
\end{equation}
Additionally, it is well-known that Euclidean distance $\norm{x-y}$ is a good approximation to $d_{M, g}(x, y)$ when $y$ is near $x$ (see, e.g., Lemma 3.1 in \cite{kde_submanifold}). That is,
\begin{equation}\label{eq:eucdist_approx}
    \frac{d_{M, g}(x, y)}{\norm{x - y}} \to 1 \qquad{\text{as }} x \to y\,.
\end{equation}
Together, Equations \ref{eq:eucdist_approx} and \ref{eq:conformaldist_approx} imply that
\begin{equation}\label{eq:dist_approx}
    \frac{d_{M, \g}(x, y)}{\sqrt[n]{\alpha(N)f(x)}{\norm{x - y}}} \to 1 \qquad{\text{as } x \to y.}
\end{equation}
We note that because $f$ is smooth, we can replace $f(x)$ in Equation \ref{eq:dist_approx} by $\max\{f(x), f(y)\}$ or by anything that converges to $f(x)$ as $x \to y$. In step 2, it is crucial to use the maximum of $\hat{f}_N(x_i)$ and $\hat{f}_N(x_j)$, rather than simply $\hat{f}_N(x)$ or some average of $\hat{f}_N(x)$ and $\hat{f}_N(y)$, so that the construction is robust with respect to outliers. Otherwise, even a single outlier in a low-density region could be deadly. The density at an outlier is very low, so the distance from an outlier to a point on the manifold would be underestimated if we did not use the maximum. For the same reason, using the maximum is also important when there are regions of differing density that are close in Euclidean space. In Sections \ref{sec:outliers} and \ref{sec:clustering}, we empirically test our method on point clouds with those challenges.

A good choice of $k$ (if such a $k$ exists) is the smallest $k$ such that two points $x$, $x'$ are in the same component of $G_{kNN}(X)$ if and only if $x$ and $x'$ are in the same component of $M$ and such that points that are ``close'' in $M$ are connected by an edge in $G_{kNN}(X)$. Heuristically, one can choose $k$ to be the first $k$ for which the number of connected components in $G_{k'NN}(X)$ is equal to the number of connected components in $G_{kNN}(X)$ for all $k' \in \{k - \ell, \cdots, k\}$ for some fixed $\ell \in \mathbb{N}$. In our experiments (Section \ref{sec:compare}), we find that $\ell = 5$ works well. Generally, it is better for $k$ to be too large than too small. A choice of $k$ that is too small could drastically change the Riemannian distance estimates if two points in the same component of $M$ are not connected in $G_{kNN}(X)$. A small value of $k$ can result in issues even when the components of $G_{kNN}(X)$ align correctly with the components of $M$. For example, if $M$ is a connected curve, then two consecutive points in $X$ on the curve may not necessarily be connected by an edge in $G_{kNN}(X)$ even if $G_{kNN}(X)$ is a connected graph.


\subsection{Definition of $\dvr$}\label{sec:dvr_imp}
\begin{definition}
Let $X = \{x_1, \ldots, x_N\}$ be a point cloud sampled from an unknown manifold of known dimension $n$. For fixed parameter $k$, the set of simplices in the \emph{approximate density-scaled Vietoris--Rips complex} $\dvr(n, k, X)$ at filtration level $t$ is
\begin{equation*}
    \Big\{ x_J \mid \widehat{d_{M, \g}}(x_i, x_j) \leq 2t \text{ for all } i, j \in J \text{ and all } J \subseteq \{1, \ldots, N\} \},
\end{equation*}
where $\widehat{d_{M, \g}}(x_i, x_j)$ is calculated as in Section \ref{sec:geo_est}.
\end{definition}
We recommend that the parameter $k$ is set by the heuristic described in Section \ref{sec:geo_est}.


\section{Stability}\label{sec:stability}
Let $X$ and $Y$ be point clouds that consist of $N$ points sampled from a smooth probability density function $f: M \to (0, \infty)$, and let $\epsilon > 0$.
We show that if $X$ and $Y$ are sufficiently close with respect to a suitable metric, then the pairwise bottleneck distances between the pairs of diagrams $\Big(\dgm(\dc(M, g, f, X))$, $\dgm(\dc(M, g, f, Y))\Big)$, $\Big(\dgm(\textnormal{DVR}(M, g, f, X))$, $\dgm(\textnormal{DVR}(M, g, f, Y))\Big)$, and $\Big(\dgm(\dvr(n, k, X))$, $\dgm(\dvr(n, k,  Y))\Big)$ are at most $\epsilon$. (For the case of $\dvr$, the point cloud $X$ must satisfy an additional constraint that is almost surely satisfied; see Definition \ref{def:general}.) By Theorem \ref{thm:interleave_bottleneck} (a result from \cite{stability}), it suffices to show that the respective complexes are \emph{$\epsilon$-interleaved}, a concept that we review below. For more details and examples of $\epsilon$-interleaving, see \cite{pers_modules}.

Let $\mathbb{U}$ and $\mathbb{V}$ be persistence modules over $\mathbb{R}$, and let $\{u_s^t: \mathbb{U}_s \to \mathbb{U}_t\}$ and $\{v_s^t: \mathbb{V}_s \to \mathbb{V}_t\}$ be the respective collections of linear maps from the persistence module structure. A \emph{homomorphism of degree $\epsilon$} is a collection $\bm{\Phi}$ of linear maps
\begin{equation*}
    \phi_t : \mathbb{U}_t \to \mathbb{V}_{t + \epsilon} \qquad \text{for all } t \in \mathbb{R}
\end{equation*}
such that the diagram
\begin{center}
\begin{tikzcd}
\mathbb{U}_s \arrow[r, "u_s^t"] \arrow[d, "\phi_s"']                   & \mathbb{U}_t \arrow[d, "\phi_t"] \\
\mathbb{V}_{s + \epsilon} \arrow[r, "v_{s + \epsilon}^{t + \epsilon}"] & \mathbb{V}_{t + \epsilon}       
\end{tikzcd}
\end{center}
commutes for all $s \leq t$. Let $\Hom^{\epsilon}(\mathbb{U}, \mathbb{V})$ denote the set of homomorphisms $\mathbb{U} \to \mathbb{V}$ of degree $\epsilon$ and $\text{End}^{\epsilon}(\mathbb{V})$ denote the set of degree-$\epsilon$ homomorphisms $\mathbb{V} \to \mathbb{V}$. A particularly important degree-$\epsilon$ endomorphism is $\bm{1}_{\mathbb{V}}^{\epsilon} \in \text{End}^{\epsilon}(\mathbb{V})$, which is the collection of maps $v_t^{t + \epsilon}$ for all $t  \in \mathbb{R}$. Composition of shifted homomorphisms is given by composition of the linear maps. If $\bm{\Phi} \in \Hom^{\epsilon_1}(\mathbb{U}, \mathbb{W})$ and $\bm{\Psi} \in \Hom^{\epsilon_2}(\mathbb{W}, \mathbb{V})$, then $\bm{\Psi}\bm{\Phi} \in \Hom^{\epsilon_1 + \epsilon_2}(\mathbb{U}, \mathbb{V})$ is the collection of linear maps
\begin{equation*}
    \psi_{t + \epsilon_1} \circ \phi_t : \mathbb{U}_t \to \mathbb{V}_{t + \epsilon_1 + \epsilon_2} \qquad \text{for all } t \in \mathbb{R}.
\end{equation*}

Persistence modules $\mathbb{U}$ and $\mathbb{V}$ are \emph{$\epsilon$-interleaved} if there are $\bm{\Phi} \in \Hom^{\epsilon}(\mathbb{U}, \mathbb{V})$ and $\bm{\Psi} \in \Hom^{\epsilon}(\mathbb{V}, \mathbb{U})$ such that
\begin{equation*}
    \bm{\Psi}\bm{\Phi} = 1_{\mathbb{U}}^{2 \epsilon}\,, \qquad \bm{\Psi}\bm{\Phi} = 1_{\mathbb{V}}^{2 \epsilon}.
\end{equation*}
A persistence module $\mathbb{V}$ is \emph{q-tame} if rank$(v_s^t) < \infty$ for all $s \leq t$. The following theorem says that the persistence diagrams of $q$-tame persistence modules that are $\epsilon$-interleaved have bottleneck distance at most $\epsilon$. We note that if $\K_r$ is a finite complex for all $r$, then its persistent homology $H(\K)$ is q-tame because $H(\K_r)$ is finite-dimensional for all $r$. Thus Theorem \ref{thm:interleave_bottleneck} applies to the persistent homology of the density-scaled complexes, which are finite at all filtration levels $r$.
\begin{theorem}[\cite{pers_modules}]\label{thm:interleave_bottleneck} If $\mathbb{U}$ and $\mathbb{V}$ are $q$-tame persistence modules that are $\epsilon$-interleaved, then the bottleneck distance between the persistence diagrams satisfies
\begin{equation*}
    W_{\infty}(\dgm(\mathbb{U}), \dgm(\mathbb{V})) \leq \epsilon\,.
\end{equation*}
\end{theorem}

\subsection{Stability of $\textnormal{DVR}$ and $\dc$}\label{sec:stability_M}
The density-scaled \v{C}ech and Vietoris--Rips complexes are defined to be the \v{C}ech and Vietoris--Rips complexes, respectively, in the density-scaled manifold $(M, \g)$. The stability properties of $\dc$ and $\textnormal{DVR}$ therefore follow from the usual stability properties of the \v{C}ech complex and the Vietoris--Rips complex \cite{stability}. Let $d_H(X, Y, (M, \g))$ denote the Hausdorff distance in $(M, \g)$ between two point clouds $X$ and $Y$ that each consist of $N$ points sampled from $f: M \to (0, \infty)$.

\begin{theorem}[Lemma 4.3 in \cite{stability}]
If $d_H(X, Y, (M, \g)) < \epsilon$, then $H(\textnormal{DVR}(M, g, f, X))$ and $H(\textnormal{DVR}(M, g, f, Y))$ are $\epsilon$-interleaved\footnote{In \cite{stability}, the Vietoris--Rips complex is defined so that there is an edge between $x$ and $y$ at filtration level $r$ if $d(x, y) \leq r$. In this paper, we use $2r$ instead. The condition from \cite{stability} on the Hausdorff distance is adjusted accordingly.}.
\end{theorem}

\begin{theorem}[Corollary 4.10 in \cite{stability}]
If $d_H(X, Y, (M, \g)) < \epsilon$, then $H(\dc(M, g, f, X))$ and $H(\dc(M, g, f, Y))$ are $\epsilon$-interleaved.
\end{theorem}

\subsection{Stability of $\dvr$}\label{sec:stability_euc}
Let $X$ and $Y$ be two point clouds that consist of $N$ points each. Their \emph{Wasserstein distance} is defined to be
\begin{equation*}
    W_{\infty}(X, Y) := \inf_{\eta: X \to Y} \norm{\eta}_{\infty},
\end{equation*}
where the infimum is taken over all bijections $\eta: X \to Y$ and where $\norm{\eta}_{\infty}:= \max_{x \in X} \norm{x - \eta(x)}$ is defined to be the largest perturbation of any point. In Theorem \ref{thm:dvr_stability}, we show that if $X$ and $Y$ are two point clouds of the same size that are in ``general position'' (defined below in Definition \ref{def:general}) and sufficiently close in Wasserstein distance, then $H(\dvr(n, k, X))$ and $H(\dvr(n, k, Y))$ are $\epsilon$-interleaved. 

First, we review a result from \cite{stability} that we use in our proof of stability.
\begin{definition}Let $\mathbb{S}$ and $\mathbb{T}$ be filtered complexes with vertices $X$ and $Y$, respectively. A bijection $\eta: X \to Y$ is \emph{$\epsilon$-simplicial} if for all $t$ and all simplices $x_J \in \mathbb{S}_t$, we have that $\eta(x_J)$ is a simplex in $\mathbb{T}_{t + \epsilon}$.
\end{definition}
The following proposition is proved in \cite{stability} in greater generality for correspondences $C:X \rightrightarrows Y$. We state the proposition for the special case in which $C$ induces a bijection.
\begin{proposition}[Proposition 4.2 in \cite{stability}]\label{prop:interleave}
If $\mathbb{S}$ and $\mathbb{T}$ are filtered complexes with vertices $X$ and $Y$, respectively, and $\eta: X \to Y$ is a bijection such that $\eta$ and $\eta^{-1}$ are both $\epsilon$-simplicial, then the persistence modules $H(\mathbb{S})$ and $H(\mathbb{T})$ are $\epsilon$-interleaved.
\end{proposition}

To prove our stability theorem for $\dvr$ (Theorem \ref{thm:dvr_stability}), we first prove a stability lemma for the density estimate (Lemma \ref{lem:hatf_stability}) and a stability lemma for the Riemannian-distance estimates (Lemma \ref{lem:geodist_stability}).

\begin{lemma}[Stability of $\hat{f}$]\label{lem:hatf_stability}
Let $X$ and $Y$ be point clouds that consist of $N$ points each, and let $\hat{f}_{X}$, $\hat{f}_{Y}$ denote the respective density estimators, as defined by Equation \ref{eq:fest}, for some kernel $K$ that satisfies the conditions of Theorem \ref{thm:kde}. For any $\epsilon > 0$, there is a $\delta > 0$ such that if $\eta: X \to Y$ is a bijection and $\norm{\eta}_{\infty} < \delta$, then $\abs{\hat{f}_X(x) - \hat{f}_Y(\eta(x))} < \epsilon$ for all $x \in X$. The value of $\delta$ depends only on the number of points $N$, the kernel $K$ used by the density estimators, and the dimension $n$.
\end{lemma}
\begin{proof}
The conditions of Theorem \ref{thm:kde} imply that $K$ is uniformly continuous. There is a $\delta > 0$ such that if $\abs{a-b} < 2\delta$, then $\abs{K(a/h_N) - K(b/h_N)} < h_N^n\epsilon$. The value of $\delta$ depends only on the kernel $K$, the dimension $n$, and the bandwidth $h_N = N^{\frac{-1}{n+4}}$. If $\norm{\eta}_{\infty} < \delta$, then
\begin{equation*}
    \abs{\hat{f}_X(x) - \hat{f}_Y(\eta(x))} \leq \frac{1}{Nh_N^n} \sum_{z \in X}\abs{K\Bigg(\frac{\norm{z - x}}{h_N}\Bigg) - K\Bigg(\frac{\norm{\eta(z) - \eta(x)}}{h_N}\Bigg)} < \epsilon.
\end{equation*}
\end{proof}

\begin{definition}\label{def:general}
We say that $X$ is in \emph{general position} with respect to parameter $k$ if every $x \in X$ has a unique set of $k$-nearest neighbors. That is, for all $x \in X$, there is a subset $N_k(x) \subseteq X$ of size $k$ such that $\norm{x - u} < \norm{x - v}$ for all $u \in N_k(x)$ and $v \in X\setminus N_k(x)$.
\end{definition}
A point cloud $X$ is in general position with respect to all $k$ whenever $\norm{x - y} \neq \norm{x - z}$ for all $x, y, z \in X$. If $X$ is a finite point cloud sampled randomly from a smooth probability density function, then $X$ is almost surely in general position for all $k$ and $X$ is always in general position with respect to $k = |X|$.

\begin{lemma}[Stability of Riemannian-Distance Estimate]\label{lem:geodist_stability}
Let $X$ and $Y$ be point clouds that consist of $N$ points each, and let $\hat{f}_X$ and $\hat{f}_Y$ denote the respective density estimators, as defined by Equation \ref{eq:fest}, for some kernel $K$ that satisfies the conditions of Theorem \ref{thm:kde}. For any $\epsilon > 0$, there is a $\delta > 0$ such that if $\eta: X \to Y$ is a bijection, $\norm{\eta}_{\infty} < \delta$, and $X$ is in general position with respect to $k$, then
\begin{equation*}
    \abs{\widehat{d_{M, \g}}(x, x') - \widehat{d_{M, \g}}(\eta(x), \eta(x'))} < \epsilon
\end{equation*}
for all $x, x' \in X$, where $\widehat{d_{M, \g}}$ is the estimate of Riemannian distance in $(M, \g)$ that is defined in Section \ref{sec:geo_est}. The value of $\delta$ depends on the point cloud $X$, the number of points $N$, the kernel $K$ used by the density estimators, and the dimension $n$.
\end{lemma}
\begin{proof}
Let $G_{kNN}(X)$ and $G_{kNN}(Y)$ be the weighted $k$-nearest neighbor graphs defined in Section \ref{sec:geo_est}, and let $N_k(x)$ and $N_k(y)$ denote the sets of $k$-nearest neighbors of $x \in X$ and $y \in Y$, respectively. There is a $\delta' > 0$ such that if  $\norm{\eta}_{\infty}< \delta'$ and $X$ is in general position with respect to $k$, then $Y$ is in general position with respect to $k$ and $N_k(\eta(x)) = \eta(N_k(x))$ for all $x$. Thus $\eta$ induces an isomorphism between the underlying \emph{unweighted} $k$-nearest neighbor graphs.

Let $w$ denote the weight function on the edges of $G_{kNN}(X)$ and $G_{kNN}(Y)$. The conditions on $K$ imply that $K$ is bounded above by some $B$, so $\sup_{z \in \mathbb{R}^m}(\hat{f}_Y(z)) \leq B/h_N^n$. By Lemma \ref{lem:hatf_stability}, there is a $\delta''$ such that if $\norm{\eta}_{\infty} < \delta''$, then
\begin{equation*}
    \abs{\sqrt[n]{\max\{\hat{f}(x_i), \hat{f}(x_j)\}} - \sqrt[n]{\max\{\hat{f}(\eta(x_i)), \hat{f}(\eta(x_j))\}}} < \frac{\epsilon}{2(N-1)\diam(X)\sqrt[n]{\alpha(N)}} \qquad{\text{for all } x_i, x_j \in X}\,.
\end{equation*}
Let $\delta = \min\{\delta', \delta'', h_N^n/[4(N-1)B\sqrt[n]{\alpha(N)}]\}$. If $\norm{\eta}_{\infty} < \delta$, then the difference between the weight of $(x_i, x_j)$ in $G_{kNN}(X)$ and the weight of $(\eta(x_i), \eta(x_j))$ in $G_{kNN}(Y)$ satisfies
\begin{align}
    \abs{w(x_i, x_j) - w(\eta(x_i), \eta(x_j))} &\leq \Big( \sqrt[n]{\alpha(N)} \norm{x_i-x_j} \abs{\sqrt[n]{\max\{\hat{f}(x_i), \hat{f}(x_j)\}} - \sqrt[n]{\max\{\hat{f}(\eta(x_i)), \hat{f}(\eta(x_j))\}}} \notag \\
    &\qquad +\sqrt[n]{\alpha(N)} \sqrt[n]{\max\{\hat{f}(\eta(x_i), \hat{f}(\eta(x_j))\}} \abs{\norm{x_i - x_j} - \norm{\eta(x_i) - \eta(x_j)}}\Big) \notag \\
    &< \frac{\epsilon}{2(N-1)} + \sqrt[n]{\alpha(N)}\frac{B}{h_N^n}2 \norm{\eta}_{\infty} \notag \\
    &< \frac{\epsilon}{N-1} \label{eq:weight_bound}.
\end{align}

If $\gamma$ is a path in $G_{kNN}(X)$ with at most $N-1$ edges, then the difference between the weighted lengths of $\gamma$ and $\eta(\gamma)$ satisfies
\begin{equation}\label{eq:pathlengthbound}
    \abs{\text{length}(\gamma) - \text{length}(\eta(\gamma))} < \epsilon
\end{equation}
by Equation \ref{eq:weight_bound}. The shortest weighted path between any two vertices has at most $N-1$ edges because if $\gamma$ is a path in $G_{kNN}(X)$ or $G_{kNN}(Y)$ with at least $N$ edges, then it must contain a cycle, and removing the cycle would create a shorter path. It follows from Equation \ref{eq:pathlengthbound} that
\begin{equation*}
    \abs{\widehat{d_{M, \g}}(x_i, x_j) - \widehat{d_{M, \g}}(\eta(x_i), \eta(x_j))} < \epsilon
\end{equation*}
for all $x_i$, $x_j \in X$.
\end{proof}

\begin{theorem}\label{thm:dvr_stability}
Let $X$ and $Y$ be point clouds that consist of $N$ points each. For any $\epsilon > 0$, there is a $\delta > 0$ such that if  $W_{\infty}(X, Y) < \delta$ and $X$ is in general position with respect to $k$, then $H(\dvr(n, k, X))$ and $H(\dvr(n, k, Y))$ are $\epsilon$-interleaved. The value of $\delta$ depends on the point cloud $X$, the number of points $N$, the kernel $K$ used by the density estimator, and the dimension $n$.
\end{theorem}
\begin{proof}
Choose $\delta$ as in the statement of Lemma \ref{lem:geodist_stability} for $2\epsilon$. 
If $W_{\infty}(X, Y) < \delta$, then there is a bijection $\eta: X \to Y$ such that $\norm{\eta}_{\infty} < \delta$. Suppose $x_J \in \dvr(n, k, X)_t$. For any $y, y' \in \eta(x_J)$, we have
\begin{equation*}
    \widehat{d_{M, \g}}(y, y') \leq \widehat{d_{M, \g}}(\eta^{-1}(y), \eta^{-1}(y')) + 2\epsilon \leq 2(t + \epsilon).
\end{equation*}
Thus $\eta(x_J) \in \dvr(n, k, Y)_{t + \epsilon}$, so $\eta$ is $\epsilon$-simplicial. By an analogous argument, $\eta^{-1}$ is $\epsilon$-simplicial. The theorem then follows from Prop \ref{prop:interleave}.
\end{proof}


\section{Empirical Performance}\label{sec:compare}

\subsection{Two Circles of Different Radii}\label{sec:circles}

We return to our example in the introduction, the point cloud in Figure \ref{fig:twocircles}, in which we sample $N = 500$ points from the disjoint union of two circles $C_1$ and $C_2$ of radius $R_1 = 1$ and $R_2 = 5$, respectively. The dimension of the manifold is $n = 1$. The probability density function is
\begin{equation*}
    f(x) = \begin{cases}
        \frac{1}{4\pi R_1} \,, & x \in C_1 \\
        \frac{1}{4\pi R_2} \,, & x \in C_2.
    \end{cases}
\end{equation*}
We estimate the density at each point by using the kernel density estimator defined in Equation \ref{eq:fest}, with the biweight kernel, and we compute shortest paths in the weighted $k$-nearest neighbor graph $G_{kNN}(X)$ to estimate Riemannian distances. To choose the parameter $k$, we follow the heuristic outlined in Section \ref{sec:geo_est}. For increasing $k$, we calculate the number of connected components in $G_{kNN}(X)$. We show the results in Figure \ref{fig:Gknn_num_comps} for $k \in \{1, \ldots, 11\}$. The number of components decreases from $k = 1$ to $k = 5$, and then remains constant at two for $k \in \{5, \ldots, 74\}$. Therefore, we set $k = 5 + \ell = 10$. (See Section \ref{sec:geo_est} for a definition of $\ell$.) In this example, the connected components of $G_{kNN}(X)$ correspond exactly to the connected components of the manifold $M$.

\begin{figure}
    \centering
    \includegraphics[scale=.75]{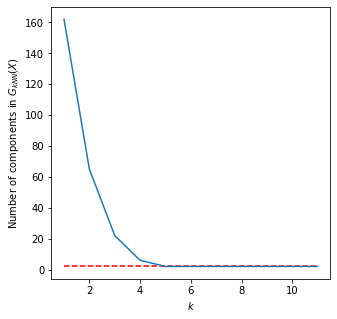}
    \caption{The number of components in $G_{kNN}(X)$ for $k \in \{1, \ldots, 11\}$, where $X$ is the point cloud in Figure \ref{fig:twocircles} sampled from the disjoint union of two circles of different radii. For $k  \in \{5, \ldots, 74\}$, the number of components is two, which is equal to the number of connected components in the manifold that we sampled $X$ from.}
    \label{fig:Gknn_num_comps}
\end{figure}
    We show the persistence diagram for $H(\dvr(n, k, X))$, with parameters $n = 1$ and $k = 10$, in Figure \ref{fig:two_circles_dvr}. Each circle has equally high resolution in the point cloud (i.e., $f(x)\sigma(x)$ is a constant function), so we expect that the lifetimes of the corresponding homology classes are equal. Indeed, the two most-persistent 1D homology classes in $H(\dvr(n, k, X))$ have lifetimes that are much closer in length than in $H(\textnormal{VR}(X))$, whose persistence diagram is shown in Figure \ref{fig:two_circles_vr}. In $H(\dvr(n, k, X))$, the two most-persistent 1D homology classes have coordinates $(0.171, 2.306)$ and $(0.130, 1.537)$, respectively; the less-persistent class has a lifetime that is $65.9\%$ the lifetime of the most-persistent class.
    By contrast, the two most persistent 1D homology classes in $H(\textnormal{VR}(X))$ have coordinates $(0.886, 8.665)$ and $(0.110, 1.733)$, respectively; the less-persistent class has a lifetime that is only $20.9\%$ the lifetime of the most-persistent class.
    One might incorrectly deduce from this that the 1D homology class with the shorter lifetime is merely noise, even though we have an equal amount of ``information'' about both circles.
    The two infinite 0D homology classes in $H(\dvr(n, k, X))$ correspond to the two connected components of the underlying manifold. The persistence diagram for $H(\textnormal{VR}(X))$ has two 0D homology classes that are significantly more persistent than the others, but only one is infinite.
\begin{figure}
    \centering
    \subfloat[]{\includegraphics[width = .5\textwidth]{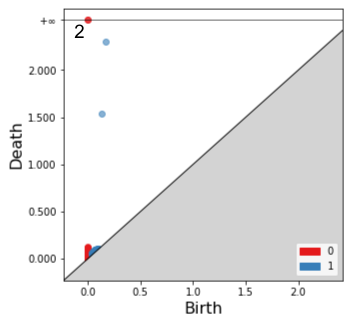}\label{fig:two_circles_dvr}}
    \subfloat[]{\includegraphics[width = .5\textwidth]{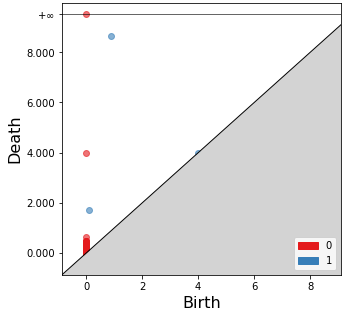}\label{fig:two_circles_vr}}
    \caption{Let $X$ be the point cloud in Figure \ref{fig:twocircles}, sampled from the disjoint union of two circles with different radii. (a) The persistence diagram for $H(\dvr(n, k, X))$, with parameters $n = 1$ and $k = 10$. The points are labeled with their multiplicity. (b) The persistence diagram for $H(\textnormal{VR}(X))$.}
\end{figure}

We also test the choice of kernel and the choice of $k$. The other two kernels that we test are the Epanechnikov kernel and the triweight kernel, and the other two values of $k$ that we test are $k=5$ and $k=15$. In Figure \ref{fig:other_choices}, we show the persistence diagrams, and in Table \ref{table:choices}, we summarize the most important features of the persistence diagrams. As we did above, we count the number of infinite 0D homology classes and we calculate the ratio of the lifetimes of the two most-persistent 1D homology classes (where a ratio closer to 1 is better). The triweight kernel with $k=10$ yields the highest ratio ($.678$), slightly higher than the ratio for the biweight kernel with $k=10$. Using a value of $k=5$ (with the biweight kernel) leads to very poor results (a ratio of $.011$) because $k = 5$ is too low for all of the adjacent points in $X$ on the largest circle to be connected by an edge in $G_{kNN}(X)$.

\begin{figure}
    \centering
    \subfloat[Epanechnikov kernel \\ $k = 10$]{\includegraphics[width = .5\textwidth]{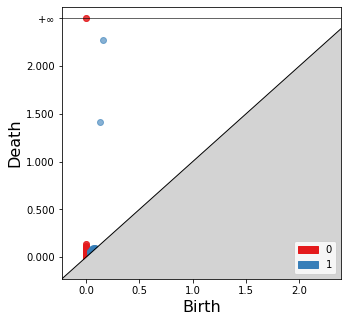}\label{fig:epan}} 
    \subfloat[Triweight kernel \\ $k = 10$]{\includegraphics[width = .5\textwidth]{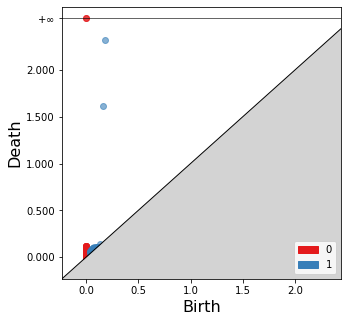}\label{fig:triweight}} \\
    \subfloat[Biweight kernel \\ $k = 5$]{\includegraphics[width = .5\textwidth]{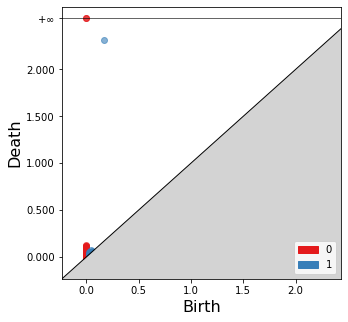}\label{fig:smallk}}
    \subfloat[Biweight kernel \\ $k = 15$]{\includegraphics[width = .5\textwidth]{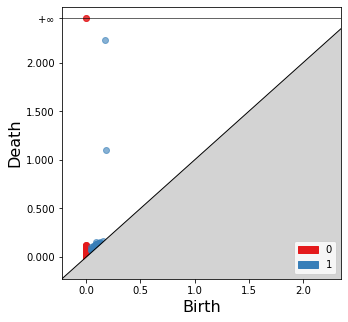}\label{fig:bigk}}
    \caption{Let $X$ be the point cloud in Figure \ref{fig:twocircles}, sampled from the disjoint union of two circles with different radii. Each plot is a persistence diagram for $H(\dvr(n, k, X))$, where we vary the value of $k$ and the choice of kernel for density estimation. The dimension is $n = 1$. In all persistence diagrams, the 0D homology class with an infinite death time has multiplicity 2.}
    \label{fig:other_choices}
\end{figure}

\begin{table}[ht]
\centering
\caption{Comparison of Kernel Functions and $k$ Values}
\begin{tabular}{||c c c c||} 
 \hline
 $k$ & Kernel function & $\frac{\text{Lifetime of second-most persistent 1D homology class}}{\text{Lifetime of most-persistent 1D homology class}}$ & Number of infinite 0D homology classes \\ [0.5ex] 
 \hline\hline
 10 & Biweight & .659 & 2\\ 
 \hline
 10 & Epanechnikov & .604 & 2\\
 \hline
 10 & Triweight & .678 & 2\\
 \hline
 5 & Biweight & .011 & 2\\
 \hline
 15 & Biweight & .442 & 2\\ [1ex] 
 \hline
\end{tabular}
\label{table:choices}
\end{table}

\subsection{Cassini Curve}
Our second example illustrates the power of the conformal-invariance property. We sample a point cloud from a Cassini curve \cite{cassini_book}, which is homeomorphic to $\mathbb{S}^1$. In polar coordinates, the equation for our Cassini curve is
\begin{equation}\label{eq:cassini}
r^4 - 2r^2 \cos2\theta = e^4 - 1\,,
\end{equation}
where $e$ is the eccentricity; we set $e = 1.01$. We sample $\theta$ uniformly at random from $[0, 2 \pi)$ and map $\theta \mapsto (\theta, r(\theta))$ as defined by Equation \ref{eq:cassini}. This is a conformal mapping from $\mathbb{S}^1$ to the Cassini curve. The mapping pushes the points $\theta = \pi/2$ and $\theta = 3\pi/2$ on the circle much closer to each other, thus decreasing the local feature size (as defined in \cite{feature_size} and in the introduction) near those points. However, the mapping also increases the density near those same points, so the mapping preserves a high level of resolution of the manifold at every point.

In Figure \ref{fig:cassini_sample}, we show a point cloud $X$ that consists of $N = 200$ points sampled from the Cassini curve. We estimate the density at each point by using the kernel density estimator defined in Equation \ref{eq:fest}, with the biweight kernel. Following the heuristic from Section \ref{sec:geo_est}, we choose $k = 12$, for which $G_{kNN}(X)$ is a connected graph (just as the Cassini curve is connected).

\begin{figure}
    \centering
    \includegraphics[scale=.8]{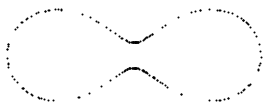}
    \caption{A point cloud $X$ that consists of $N = 200$ points sampled from the Cassini curve defined in polar coordinates $(\theta, r)$ by Equation \ref{eq:cassini}. The angle $\theta$ is sampled uniformly at random.
    }
    \label{fig:cassini_sample}
\end{figure}

We show the persistence diagram for $H(\dvr(n, k, X))$ (with parameters $n = 1$ and $k = 12$) in Figure \ref{fig:cassini_dvr}. The diagram shows only one infinite 0D homology class and one 1D homology class, which is consistent with the homology of the Cassini curve because the Cassini curve has one 0D homology generator (it's connected) and one 1D homology generator. We compare to the persistence diagram for $H(\textnormal{VR}(X))$, shown in Figure \ref{fig:cassini_vr}. The diagram shows two nearly equally-persistent 1D homology classes, which is not consistent with the topology of the Cassini curve. 

\begin{figure}
    \centering
    \subfloat[]{\includegraphics[width = .5\textwidth]{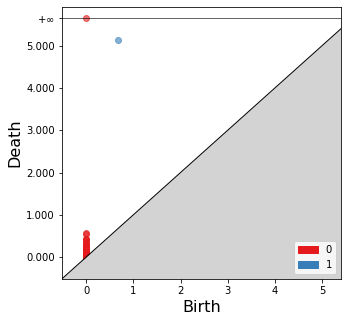}\label{fig:cassini_dvr}}
    \subfloat[]{\includegraphics[width = .5\textwidth]{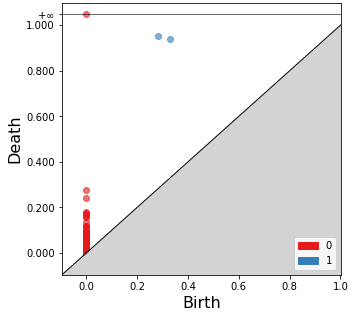}\label{fig:cassini_vr}}
    \caption{Let $X$ be the point cloud in Figure \ref{fig:cassini_sample}, sampled from the Cassini curve. (a) The persistence diagram for $H(\dvr(n, k, X))$, with parameters $n = 1$ and $k = 12$. (b) The persistence diagram for $H(\textnormal{VR}(X))$.}
\end{figure}


\subsection{A Point Cloud with Outliers}\label{sec:outliers}
In this subsection, we evaluate the performance of $\dvr$ on a point cloud that contains outliers. Because outliers lie in low-density regions, where density-scaled Riemannian distances are much shorter than Euclidean distances, one might be concerned that outliers that are not on the manifold would get connected to points on the manifold too quickly. Indeed, this is what happens in the density-weighted complexes of Appendix \ref{sec:weighted}. (We discuss this at the end of the present subsection.) Here, we show that $\dvr$ does not suffer from this problem, and in fact $\dvr$ outperforms the Vietoris--Rips complex on the point cloud that we test on.

The point cloud $X$ in Figure \ref{fig:noisy_circle} consists of $200$ points sampled uniformly at random from the manifold $M = \mathbb{S}^1$ and $10$ points sampled uniformly at random from the square $[-1, 1]^2$. The dimension of the manifold is $n = 1$, and we set $k = 10$ by following the heuristic in Section \ref{sec:geo_est}. We estimate the density at each point by using the kernel density estimator defined in Equation \ref{eq:fest}, with the biweight kernel.

\begin{figure}
    \centering
    \includegraphics[scale=.5]{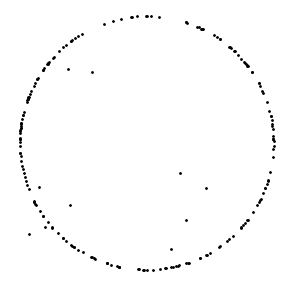}
    \caption{A point cloud $X$ that consists of $200$ points drawn uniformly at random from $\mathbb{S}^1$ and $10$ points drawn uniformly at random from $[-1, 1]^2$.
    }
    \label{fig:noisy_circle}
\end{figure}

In Figure \ref{fig:noisy_ph}, we show the persistence diagrams for $H(\dvr(n, k, X))$ and $H(\textnormal{VR}(X))$. Both diagrams have one 1D homology class that has a much longer lifetime than the others. This correctly reflects the topology of $\mathbb{S}^1$, which has one generator for its 1D homology. The lifetime of the most-persistent 1D homology class is much longer in $H(\dvr(n, k, X))$ than in $H(\textnormal{VR}(X))$. In $H(\dvr(n, k, X))$, the most-persistent 1D homology class is born at $t = .299$ and dies at at $t = 3.107$, whereas in $H(\textnormal{VR}(X))$, it is born at $t = 0.212$ and dies much earlier at $t = 1.102$. 

\begin{figure}
    \centering
    \subfloat[]{\includegraphics[width =.5\linewidth]{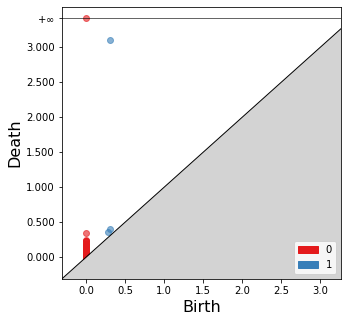}}
    \subfloat[]{\includegraphics[width = .5\linewidth]{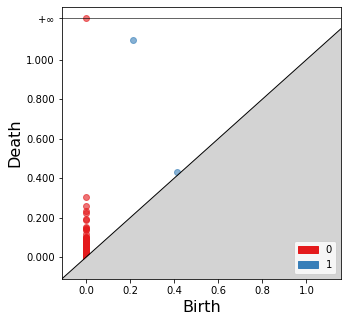}}
    \caption{Let $X$ be the point cloud in Figure \ref{fig:noisy_circle}, which consists of $200$ points sampled uniformly at random from $\mathbb{S}^1$ and $10$ points sampled uniformly at random from $[-1, 1]^2$. (a) The persistence diagram for $\dvr(n, k, X)$, with parameters $n = 1$ and $k = 10$. The 1D homology class with the longest lifetime has coordinates $(0.299, 3.107)$. (b) The persistence diagram for $H(\textnormal{VR}(X))$. The 1D homology class with the longest lifetime has coordinates $(0.212, 1.102)$.}
    \label{fig:noisy_ph}
\end{figure}

We also compare $\dvr$ to the density-weighted Vietoris--Rips complex (Appendix \ref{sec:weighted}) to show that the density-weighted Vietoris--Rips complex may not perform well when there are outliers in the point cloud. In Figure \ref{fig:noisy_1skeleton}, we show the 1-skeleton of the density-weighted complexes for increasing filtration level $t$. The outliers quickly connect to points on the circle. In Figure \ref{fig:noisy_weightedph}, we show the persistence diagram for the density-weighted Vietoris--Rips complex. The 1D persistence is a poor reflection of the 1D homology of the underlying manifold, which has a single generator.

\begin{figure}
    \centering
    \subfloat[$t = 0.05$]{\includegraphics[width = .3\linewidth]{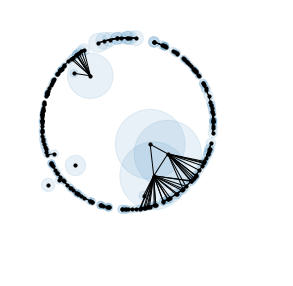}}
    \subfloat[$t = 0.1$]{\includegraphics[width = .3\linewidth]{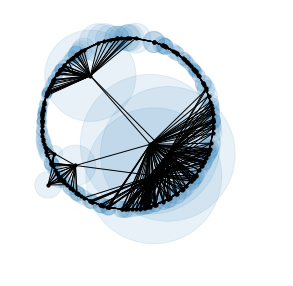}}
    \subfloat[$t = 0.15$]{\includegraphics[width = .3\linewidth]{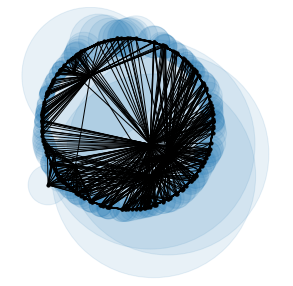}}
    \caption{Let $X$ be the point cloud in Figure \ref{fig:noisy_circle}, which consists of $200$ points sampled uniformly at random from $\mathbb{S}^1$ and $10$ points drawn uniformly at random from $[-1, 1]^2$. In each plot, the ball centered at $x \in X$ has radius $\frac{t}{\sqrt[n]{\alpha(N)\hat{f}_N(x)}}$, where $\hat{f}_N(x)$ is the estimate of $f(x)$ defined by Equation \ref{eq:fest}. We display the 1-skeleton of the density-weighted Vietoris--Rips complex (equivalently, the 1-skeleton of the density-weighted \v{C}ech complex) for $t = 0.05$, $t = 0.1$, and $t = 0.15$.}
    \label{fig:noisy_1skeleton}
\end{figure}

\begin{figure}
    \centering
    \includegraphics[width = .5\linewidth]{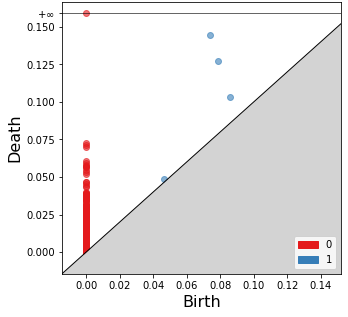}
    \caption{The persistence diagram for the persistent homology of the density-weighted Vietoris--Rips complex for the point cloud $X$ in Figure \ref{fig:noisy_circle}, which consists of $200$ points sampled uniformly at random from $\mathbb{S}^1$ and $10$ points drawn uniformly at random from $[-1, 1]^2$.}
    \label{fig:noisy_weightedph}
\end{figure}


\subsection{Application: Clustering}\label{sec:clustering}
In this example, we show that $\dvr$ is effective at identifying the number of clusters in a point cloud whose clusters have different densities. We sample $N = 200$ points from the union of the squares $[0,1]\times[0,1]$ and $[1.5, 2.5]\times [0,1]$. The probability density function is
\begin{equation}\label{eq:clusters_f}
    f(x) = \begin{cases}
        1/6 \,, & x \in [0,1]\times[0,1] \\
        5/6 \,, & x \in [1.5, 2.5]\times[0,1].
    \end{cases}
\end{equation}
We show the point cloud $X$ in Figure \ref{fig:clusters}.
\begin{figure}
    \centering
    \includegraphics{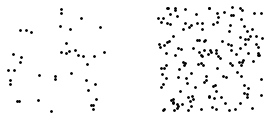}
    \caption{A point cloud $X$ that consists of $N = 200$ points sampled from the union of the squares $[0,1]^2$ and $[1.5, 2.5]\times [0,1]$ with probability density function given by Equation \ref{eq:clusters_f}.}
    \label{fig:clusters}
\end{figure}

We show the persistence diagram for $H(\dvr(n, k, X))$, with parameters $n = 2$ and $k = 9$ (chosen by the heuristic in Section \ref{sec:geo_est}) in Figure \ref{fig:clusters_dvr}. The two infinite 0D homology classes correspond to the two connected components of the underlying manifold (the two squares). In Figures \ref{fig:clusters_vr}, \ref{fig:clusters_knn}, and \ref{fig:clusters_weighted}, we compare the persistence diagram for $H(\dvr(n, k, X))$ to the persistence diagrams for the Vietoris--Rips complex, the $\textnormal{KNN}$ complex (Appendix \ref{sec:knn}), and the density-weighted Vietoris--Rips complex (Appendix \ref{sec:weighted}). None of these three persistence diagrams show two equally-persistent 0D homology classes. Moreover, the persistence diagrams for the $\textnormal{KNN}$ filtration and the density-weighted filtration show several 1D homology classes with lifetimes that are almost as long as the 0D homology class with the 2nd-longest lifetime, even though the underlying manifold has trivial 1D homology.
\begin{figure}
    \centering
    \subfloat[]{\includegraphics[width = .5\textwidth]{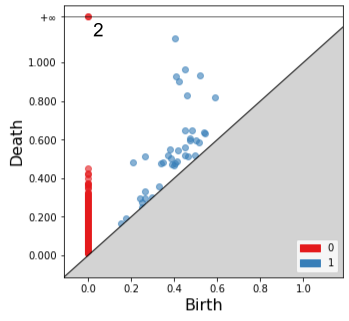}\label{fig:clusters_dvr}} 
    \subfloat[]{\includegraphics[width = .5\textwidth]{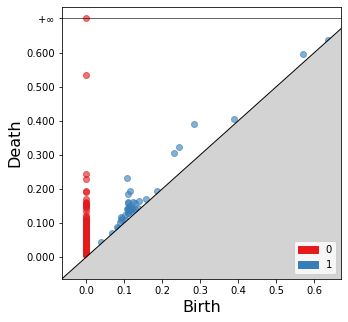} \label{fig:clusters_vr}} \\
    \subfloat[]{\includegraphics[width=.5\textwidth]{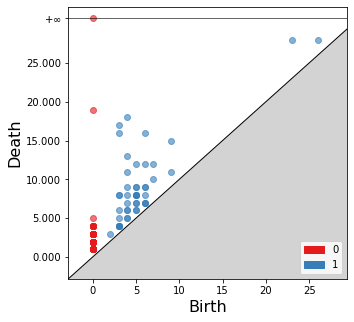}\label{fig:clusters_knn}}
    \subfloat[]{\includegraphics[width=.5\textwidth]{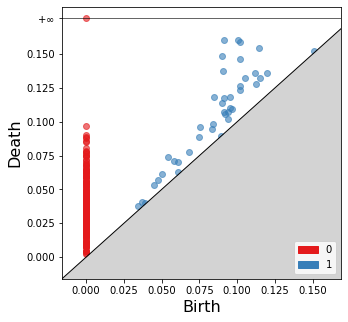}\label{fig:clusters_weighted}}
    \caption{Let $X$ be the point cloud in Figure \ref{fig:clusters} sampled from the union of the squares $[0,1]^2$ and $[1.5, 2.5]\times [0,1]$ with probability density function given by Equation \ref{eq:clusters_f}. (a) The persistence diagram for $H(\dvr(n, k, X))$, with parameters $n = 2$ and $k = 9$. The points are labeled with their multiplicity. (b) The persistence diagram for $H(\textnormal{VR}(X))$. (c) The persistence diagram for $H(\textnormal{KNN}(X))$. (d) The persistence diagram for the PH of the density-weighted Vietoris--Rips complex defined in Appendix \ref{sec:weighted}.}
\end{figure}


\subsection{Application: Lorenz System}\label{sec:lorenz}
As a final demonstration of our method, we apply $\dvr$ to a point cloud generated from the Lorenz `63 dynamical system \cite{lorenz}. The equations of motion are
\begin{equation}\label{eq:lorenz}
\begin{cases}
    \frac{\mathrm{d}x}{\mathrm{d}t} = \sigma (y - x)\,, \\
    \frac{\mathrm{d}y}{\mathrm{d}t} = x (\rho - z)-y\,, \\
    \frac{\mathrm{d}z}{\mathrm{d}t} = xy - \beta z\,,
\end{cases}
\end{equation}
where $\sigma$, $\rho$, and $\beta$ are system parameters. We set $\sigma = 10$, $\rho = 28$, and $\beta = 8/3$, which are the values that Lorenz used in \cite{lorenz}. In this subsection, we study a point cloud that is sampled from a time-delay embedding of $x(t)$. Dynamical systems are often studied via such time-delay embeddings because the image of the time-delay embedding is diffeomorphic to the attractor manifold (under suitable conditions) by Takens' embedding theorem \cite{takens}. The point cloud that results from the time-delay embedding is an interesting application for $\dvr$ because it is a point cloud that Vietoris--Rips does not perform well on.

The Lorenz attractor has two visible holes that correspond to two equilibrium points (see Figure \ref{fig:lorenz}). These holes correspond to two ``topological regimes'' in the dynamical system, as defined by \cite{top_regimes}. Observe that the two holes are of different sizes and that the density is higher near the smaller hole. The difference in density is even more pronounced in the time-delay embedding, which we show in Figure \ref{fig:lorenzshadow}.

We construct a point cloud as follows. Our initial condition for Equation \ref{eq:lorenz} is $(x_0, y_0, z_0) = (1, 1, 1)$, and we solve the system from time $t = 0$ to time $t = 50$ using the SciPy ODE solver \cite{scipy}. This results in an approximate solution $(x(t), y(t), z(t))$. In Figure \ref{fig:lorenz}, we show the collection of points $\{(x(t_i), y(t_i), z(t_i))\}_{i = 0}^{1000}$, with time steps $t_i = .05i$. We define the point cloud $X$ to be the $2$-dimensional time-delay embedding of $x(t)$ with time lag $\tau = .05$ (see Figure \ref{fig:lorenzshadow}).

In Figure \ref{fig:lorenz_ph}, we compare the persistence diagrams for $H(\dvr(n, k, X))$ (with parameters $n = 2$ and $k = 10$) and $H(\textnormal{VR}(X))$. The persistence diagram for $H(\dvr(n, k, X))$ picks up on both equilibrium points, as it has two 1D homology classes with significantly longer lifetimes than the other homology classes. By contrast, the Vietoris--Rips persistence diagram picks up on only one of the equilibrium points, as it has only one 1D homology class with a significantly longer lifetime than the others.

\begin{figure*}
    \centering
    \subfloat[]{\includegraphics[width = .47\textwidth]{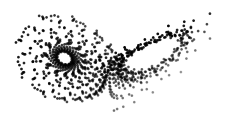}\label{fig:lorenz}}
    \hspace{5mm}
    \subfloat[]{\includegraphics[width = .47\textwidth]{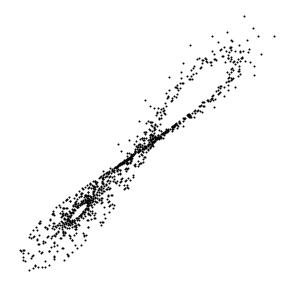}\label{fig:lorenzshadow}}
    \caption{(a) A point cloud that consists of $N = 1000$ points sampled from an approximate solution to the Lorenz 
    `63 system defined by Equation \ref{eq:lorenz}. (b) A $2$-dimensional time-delay embedding of $x(t)$.}
\end{figure*}

\begin{figure}
    \centering
    \subfloat[]{\includegraphics[width = .5\textwidth]{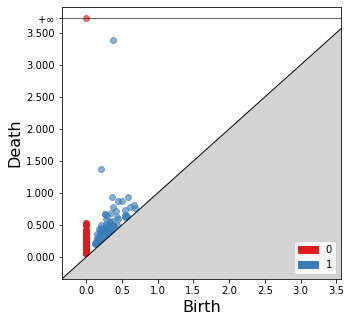}}
    \subfloat[]{\includegraphics[width = .5\textwidth]{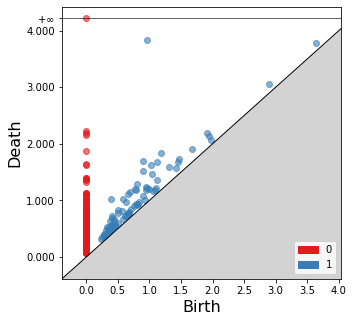}}
    \caption{Let $X$ be the point cloud in Figure \ref{fig:lorenzshadow}, which consists of $1000$ points sampled from a $2$-dimensional time-delay embedding of $x(t)$ in the Lorenz `63 dynamical system defined by Equation \ref{eq:lorenz}. (a) The persistence diagram for $H(\dvr(n, k, X))$, with parameters $n = 2$ and $k = 10$. (b) The persistence diagram for $H(\textnormal{VR}(X))$.}
    \label{fig:lorenz_ph}
\end{figure}


\section{Conclusions and Discussion}\label{sec:discussion}

In this paper we constructed a family of density-scaled filtered complexes, with a focus on the density-scaled \v{C}ech complex and the density-scaled Vietoris--Rips complex. The density-scaled filtered complexes improve our abilities to analyze point clouds with locally-varying density and to compare point clouds with differing densities. Given a Riemannian manifold $(M, g)$ and a smooth probability density function $f: M \to (0, \infty)$ from which a point cloud $X$ is sampled, we defined a conformally-equivalent density-scaled Riemannian manifold $(M, \g)$ such that $X$ is uniformly distributed in $(M, \g)$. This allowed us to define a density-scaled \v{C}ech complex $\dc$ and more generally to define a density-scaled version of any distance-based filtered complex, including a density-scaled Vietoris--Rips complex $\textnormal{DVR}$. We proved that the density-scaled \v{C}ech complex has better guarantees on topological consistency than the standard \v{C}ech complex, and we showed that the density-scaled complexes are invariant under conformal transformations (e.g., scaling), which brings topological data analysis closer to being a topological tool. By using kernel-density estimation and Riemannian-distance estimation techniques from \cite{isomap, Cisomap}, we implemented a filtered complex $\dvr$ that approximates $\textnormal{DVR}$. We compared $\dvr$ to the usual Vietoris--Rips complex and found in our experiments that $\dvr$ was better than Vietoris--Rips at providing information about the underlying homology of $M$.

Our definitions of the density-scaled complexes required the point clouds to be sampled from a Riemannian manifold with \emph{global} intrinsic dimension $n$. However, our implementation $\dvr$ immediately generalizes to metric spaces for which the intrinsic dimension varies locally. (A trivial example of such a metric space is the disjoint union of Riemannian manifolds with different dimensions.) One can estimate the local intrinsic dimension $n_x$ near a point $x$ using one of the methods of \cite{local_PCA, intrinsic_dim, conical_dimestimate, ball_dimestimate, doubling_dimestimate}, for example. In the density estimator $\hat{f}_N(y)$ defined by Equation \ref{eq:fest}, one replaces $n$ by $n_y$. For estimating Riemannian distance, one can construct the $k$-nearest neighbor graph $G_{kNN}(X)$ as usual, but with $n$ replaced by an average of $n_{x_i}$ and $n_{x_j}$ when defining the weight of the edge $(x_i, x_j)$. Thus one can compute $\dvr(X)$ for a point cloud $X$ that is sampled from a space whose intrinsic dimension varies. It would certainly be of interest to analyze the performance and theoretical guarantees of $\dvr$ on metric spaces of varying local dimension.

The implementation $\dvr$ can also be improved by improving the algorithm for Riemannian-distance estimation in $(M, \g)$. The construction of the $k$-nearest neighbor graph $G_{kNN}(X)$ is stable (assuming $X$ is in ``general position'', as in Definition \ref{def:general}), but it is much more sensitive to perturbations than one would like. Additionally, we note that sometimes all $k$-nearest neighbors of a point $x \in X$ are in the same direction\footnote{By ``the same direction'' we mean the following. Within the injectivity radius of $x$, the exponential map $\exp_x(v)$ is a diffeomorphism \cite{petersen}. Two neighbors $y_1$, $y_2$ are ``in the same direction'' if there is a $v \in T_xM$ within the injectivity radius such that $\exp_x(v) = y_1$ and $\exp_x(tv) = y_2$ for some $t > 0$. More generally, if $\exp_x(v_1) = y_1$ and $\exp_x(v_2) = y_2$, then the angle between $v_1$ and $v_2$ is a way of quantifying how close in direction they are.} relative to $x$. This leads to problems such as the following: if $M$ is a curve, then even two adjacent points $x_1$, $x_2 \in X$ on the curve may not be connected by an edge in $G_{kNN}(X)$ if the parameter $k$ is too small. (We observed this in Section \ref{sec:circles} when we tried setting $k = 5$.) A solution would be to estimate the tangent space at each point and to connect $x$ only to nearest neighbors that lie in different directions. Such a modification could also improve Riemannian-distance estimation in widely-used algorithms such as Isomap \cite{isomap}.

The density-scaled complexes we defined in this paper are conformally invariant, but it is desirable to construct filtered complexes that are invariant under a wider class of diffeomorphisms. Conformal invariance is the best that one can hope for in our setting because the density-scaled Riemannian manifold $(M, \g)$ is conformally equivalent to $(M, g)$. One idea for improving on the current definition is to consider the local covariance of the probability distribution at each point \cite{Rman_stats} and to modify the Riemannian metric in such a way that with respect to the new Riemannian metric, the local covariance matrix at each point is the identity matrix (with respect to a positively oriented orthonormal basis). This idea is akin to the usual normalization that data scientists often do in Euclidean space, and it is also reminiscent of the ellipsoid-thickenings of \cite{ellipsoids}.


\section*{Appendix}
\subsection{Other Filtrations}\label{sec:otherfilts}
\subsubsection{$k$-Nearest Neighbor Filtration}\label{sec:knn}
In the $k$-nearest neighbors filtration, each point is connected to its $k$-nearest neighbors at the $k$th filtration step.
\begin{definition}\label{def:knn}
Let $X$ be a point cloud. The $\textnormal{KNN}(X)$ complex is the filtered complex such that at filtration level $k$, the set of simplices in $\textnormal{KNN}(X)_k$ is
\begin{equation*}
    \Big\{ x_J \mid \norm{x_i - x_j} \leq \norm{x_i - x_{N^k_i}} \text{ or } \norm{x_i - x_j} \leq \norm{x_j - x_{N^k_j}} \text{for all } i, j \in J \text{ and all } J \subseteq \{1, \ldots, N\} \Big\}\,,
\end{equation*}
where $N^k_i$ is the index of the $k$th nearest neighbor of $x_i$ and $N^k_j$ is the index of the $k$th nearest neighbor of $x_j$.
\end{definition}

\subsubsection{Weighted Filtrations}\label{sec:weighted}
Our density-scaled filtered complexes are similar in spirit to the weighted \v{C}ech complex and the weighted Vietoris--Rips complex from \cite{weightedPH}. We recall the definitions here.

\begin{definition}
Let $X$ be a point cloud. Let $\mathcal{C}^1_+([0, \infty))$ denote the collection of differentiable bijective functions $\phi: [0, \infty) \to [0, \infty)$ with positive first derivative. A \emph{radius function} on $X$ is a function $r: X \to \mathcal{C}^1_+([0, \infty))$. We denote the image function $r(x) \in \mathcal{C}^1_+([0, \infty))$ by $r_x$.
\end{definition}
For example, if $\{s_x\}_{x \in X}$ is a set of positive real numbers, then $r_x(t) = ts_x$ is a radius function on $X$. This models the case in which the ball centered at $x \in X$ grows linearly over time $t$.
\begin{definition}
Let $X$ be a point cloud and let $r: X \to \mathcal{C}^1_+([0, \infty))$ be a radius function. The \emph{weighted \v{C}ech complex} at filtration level $t \geq 0$ is the nerve of $\{B(x, r_x(t))\}_{x \in X}$.
\end{definition}

\begin{definition}
Let $X$ be a point cloud and let $r: X \to \mathcal{C}^1_+([0, \infty))$ be a radius function. The set of simplices in the \emph{weighted Vietoris--Rips complex} at filtration level $t \geq 0$ is
\begin{equation*}
    \Big\{ x_J \mid \norm{x_i - x_j} \leq r_{x_i}(t) + r_{x_j}(t) \text{ and all } J \subseteq \{1, \ldots, N \}\Big\}.
\end{equation*}
\end{definition}

When $r_x(t)$ is the radius function defined in Equation \ref{eq:weighted_radius}, we refer to these as the \emph{density-weighted \v{C}ech complex} and the \emph{density-weighted Vietoris--Rips complex}, respectively.


\subsection{Riemannian Geometry Lemmas}\label{sec:lemmas}

\begin{lemma}\label{lem:volume_transform}
Let $(M_1, g_1)$ and $(M_2, g_2)$ be Riemannian manifolds with volume forms $dV_1$ and $dV_2$ respectively, and let $F: M_2 \to M_1$ be a diffeomorphism. If $F^*g_1 = \lambda g_2$ for some positive smooth function $\lambda: M_2 \to (0, \infty)$, then $F^*dV_1 = \sqrt{\lambda^n} dV_2$.
\end{lemma}
\begin{proof}
Let $\{e_1, \ldots, e_n\}$ be a positively oriented orthonormal basis for $T_y M_2$. By hypothesis,
\begin{equation*}
    g_1[F(y)]\Big(\frac{1}{\sqrt{\lambda}} dF_ye_i, \frac{1}{\sqrt{\lambda}} dF_y e_j\Big) = \frac{1}{\lambda}F^*g_1[y](e_i, e_j) = g_2[y](e_i, e_j) = \delta_{ij}.
\end{equation*}
Therefore, $\{\frac{1}{\sqrt{\lambda}} dF_ye_i\}$ is a positively oriented orthonormal basis for $T_{F(y)}M_1$, so
\begin{align*}
    \frac{1}{\sqrt{\lambda^\Mdim}}F^*dV_1[y](e_1, \ldots, e_n) &= \frac{1}{\sqrt{\lambda^\Mdim}} dV_1[F(y)](dF_ye_1, \ldots, dF_y e_n) \\
    &= dV_1[F(y)]\Big(\frac{1}{\sqrt{\lambda}}dF_ye_1, \ldots, \frac{1}{\sqrt{\lambda}}dF_ye_n\Big) \\
    &= 1.
\end{align*}
We must have $\frac{1}{\sqrt{\lambda^\Mdim}}F^*dV_1 = dV_2$ because $dV_2$ is the unique $n$-form that equals $1$ on every positively oriented orthonormal basis.
\end{proof}
\begin{lemma}\label{lem:F_induces_iso}
Let $(M_1, g_1)$ and $(M_2, g_2)$ be Riemannian manifolds, and let $f_1 : M_1 \to (0, \infty)$ be a smooth probability density function on $M_1$. If $F: (M_2, g_2) \to (M_1, g_1)$ is a conformal transformation and $f_2$ is the pullback of $f_1$ defined by Definition \ref{def:pdf_pullback}, then $F: (M_1, \g_1) \to (M_2, \g_2)$ is an isometry, where $\g_1$ and $\g_2$ are the density-scaled Riemannian metrics on $M_1$ and $M_2$, respectively.
\end{lemma}
\begin{proof}
There is a positive smooth function $\lambda: M_2 \to (0, \infty)$ such that $F^*g_1 = \lambda g_2$.
By Lemma \ref{lem:volume_transform}, $F^*dV_1 = \sqrt{\lambda^\Mdim}dV_2$. Therefore,
\begin{equation*}
    F^*(fdV_1) = (f\circ F)\sqrt{\lambda^\Mdim}dV_2,
\end{equation*}
so the pullback of $f_1$ defined by Definition \ref{def:pdf_pullback} is $f_2 = (f \circ F)\sqrt{\lambda^\Mdim}$. The density-scaled Riemannian metric on $M_2$ is 
\begin{equation*}
    \g_2 = \sqrt[\Mdim]{\alpha(N)^2(f \circ F)^2} \lambda g_2.
\end{equation*}
Therefore,
\begin{equation*}
    F^*\g_1 = \sqrt[\Mdim]{\alpha(N)^2(f \circ F)^2} F^*g_1 = \sqrt[\Mdim]{\alpha(N)^2(f \circ F)^2} \lambda g_2 = \g_2.
\end{equation*}
\end{proof}

\begin{remark}
In fact, $F^* \g_1 = \g_2$ if and \emph{only if} $F$ is a conformal transformation because $\g_1$ and $\g_2$ are conformally equivalent to $g_1$ and $g_2$, respectively.
\end{remark}

\begin{acknowledgements}
I am very grateful to Jerry Luo, Nina Otter, and Mason Porter for comments on an earlier draft. I would also like to thank Mike Hill, Peter Petersen, Luis Scoccola, and Renata Turkes for helpful discussions.
\end{acknowledgements}

%


%
\bibliographystyle{spmpsci}      
\bibliography{references}   

\begin{thebibliography}{10}
\providecommand{\url}[1]{{#1}}
\providecommand{\urlprefix}{URL }
\expandafter\ifx\csname urlstyle\endcsname\relax
  \providecommand{\doi}[1]{DOI~\discretionary{}{}{}#1}\else
  \providecommand{\doi}{DOI~\discretionary{}{}{}\begingroup
  \urlstyle{rm}\Url}\fi

\bibitem{pers_images}
Adams, H., Emerson, T., Kirby, M., Neville, R., Peterson, C., Shipman, P.,
  Chepushtanova, S., Hanson, E., Motta, F., Ziegelmeier, L.: Persistence
  images: {A} stable vector representation of persistent homology.
\newblock Journal of Machine Learning Research \textbf{18}(8), 1--35 (2017)

\bibitem{feature_size}
Amenta, N., Bern, M.: Surface reconstruction by voronoi filtering.
\newblock Discrete \& Computational Geometry \textbf{22}, 481--504 (1999)

\bibitem{laplacian_eigmap}
Belkin, M., Niyogi, P.: Laplacian eigenmaps for dimensionality reduction and
  data representation.
\newblock Neural Computation \textbf{15}(6), 1373--1396 (2003)

\bibitem{weightedPH}
Bell, G., Lawson, A., Martin, J., Rudzinski, J., Smyth, C.: Weighted persistent
  homology.
\newblock Involve: A Journal of Mathematics \textbf{12}(5), 823--837 (2019)

\bibitem{bendich2016}
Bendich, P., Marron, J.S., Miller, E., Pieloch, A., Skwerer, S.: Persistent
  homology analysis of brain artery trees.
\newblock The Annals of Applied Statistics \textbf{10}(1), 198--218 (2016)

\bibitem{convexity_radius}
Berger, M.: A Panoramic View of Riemannian Geometry, 1st edn., chap. 6.5.3.
\newblock Springer-Verlag, Berlin, Heidelberg (2003)

\bibitem{kde_sublevel}
Bobrowski, O., Mukherjee, S., Taylor, J.E.: {Topological consistency via kernel
  estimation}.
\newblock Bernoulli \textbf{23}(1), 288--328 (2017)

\bibitem{vanishing_homology}
Bobrowski, O., Weinberger, S.: On the vanishing of homology in random \v{C}ech
  complexes.
\newblock Random Structures \& Algorithms \textbf{51}(1), 14--51 (2017)

\bibitem{fermat_tda}
Borghini, E., Fernández, X., Groisman, P., Mindlin, G.: Intrinsic persistent
  homology via density-based metric learning.
\newblock arXiv:2012.07621v2  (2021)

\bibitem{Borsuk}
Borsuk, K.: On the imbedding of systems of compacta in simplicial complexes.
\newblock Fundamenta Mathematicae \textbf{35}(1), 217--234 (1948)

\bibitem{Bubenik2020}
Bubenik, P., Hull, M., Patel, D., Whittle, B.: Persistent homology detects
  curvature.
\newblock Inverse Problems \textbf{36}(2), 025008 (2020)

\bibitem{materials}
Buchet, M., Hiraoka, Y., Obayashi, I.: Persistent homology and materials
  informatics.
\newblock In: I.~Tanaka (ed.) Nanoinformatics, pp. 75--95. Springer, Singapore
  (2018)

\bibitem{pers_modules}
Chazal, F., {de Silva}, V., Glisse, M., Oudot, S.: The Structure and Stability
  of Persistence Modules, 1st edn.
\newblock Springer, Cham, Switzerland (2016)

\bibitem{stability}
Chazal, F., {de Silva}, V., Oudot, S.: Persistence stability for geometric
  complexes.
\newblock Geometriae Dedicata \textbf{173}, 193--214 (2014)

\bibitem{dtm_tda}
Chazal, F., Fasy, B., Lecci, F., Michel, B., Rinaldo, A., Wasserman, L.: Robust
  topological inference: Distance to a measure and kernel distance.
\newblock Journal of Machine Learning Research \textbf{18}(159), 1--40 (2018)

\bibitem{cortical}
Chung, M.K., Bubenik, P., Kim, P.T.: Persistence diagrams of cortical surface
  data.
\newblock In: J.L. Prince, D.L. Pham, K.J. Myers (eds.) Information Processing
  in Medical Imaging, vol. 5636, pp. 386--397. Springer, Berlin, Heidelberg
  (2009)

\bibitem{hessian_eigenmap}
Donoho, D.L., Grimes, C.: Hessian eigenmaps: Locally linear embedding
  techniques for high-dimensional data.
\newblock Proceedings of the National Academy of Sciences \textbf{100}(10),
  5591--5596 (2003)

\bibitem{edel_book}
Edelsbrunner, H., Harer, J.: Computational Topology: An Introduction.
\newblock American Mathematical Society, Providence, RI (2010)

\bibitem{feng2021}
Feng, M., Porter, M.A.: Persistent homology of geospatial data: A case study
  with voting.
\newblock SIAM Review \textbf{63}(1), 67--99 (2021)

\bibitem{flatto_newman}
Flatto, L., Newman, D.J.: Random coverings.
\newblock Acta Mathematica \textbf{138}, 241--264 (1977)

\bibitem{intrinsic_dim}
Fukunaga, K., Olsen, D.R.: An algorithm for finding intrinsic dimensionality of
  data.
\newblock IEEE Transactions on Computers \textbf{C-20}(2), 176--183 (1971)

\bibitem{eat}
Ghrist, R.: Elementary Applied Topology, 1st edn.
\newblock Createspace (2014)

\bibitem{fermat}
Groisman, P., Jonckheere, M., Sapienza, F.: Nonhomogeneous {E}uclidean
  first-passage percolation and distance learning.
\newblock arXiv:1810.09398v2  (2019)

\bibitem{doubling_dimestimate}
Gupta, A., Krauthgamer, R., Lee, J.R.: Bounded geometries, fractals, and
  low-distortion embeddings.
\newblock In: 44th Annual IEEE Symposium on Foundations of Computer Science,
  2003. Proceedings., pp. 534--543 (2003)

\bibitem{hatcher}
Hatcher, A.: {Algebraic Topology}, 1st edn.
\newblock Cambridge University Press, Cambridge, UK (2002)

\bibitem{ellipsoids}
Kalisnik, S., Lesnik, D.: Finding the homology of manifolds using ellipsoids.
\newblock arXiv:2006.09194v2  (2021)

\bibitem{local_PCA}
Kambhatla, N., Leen, T.K.: Dimension reduction by local principal component
  analysis.
\newblock Neural Computation \textbf{9}(7), 1493--1516 (1997)

\bibitem{ball_dimestimate}
Karger, D.R., Ruhl, M.: Finding nearest neighbors in growth-restricted metrics.
\newblock In: Proceedings of the Thirty-Fourth Annual ACM Symposium on Theory
  of Computing, STOC '02, pp. 741--750. Association for Computing Machinery,
  New York, NY (2002)

\bibitem{knn_density}
Loftsgaarden, D.O., Quesenberry, C.P.: A nonparametric estimate of a
  multivariate density function.
\newblock The Annals of Mathematical Statistics \textbf{36}(3), 1049--1051
  (1965)

\bibitem{lorenz}
Lorenz, E.N.: Deterministic nonperiodic flow.
\newblock Journal of Atmospheric Sciences \textbf{20}(2), 130--141 (1963)

\bibitem{cyclo}
Martin, S., Thompson, A., Coutsias, E., Watson, J.P.: Topology of cyclooctane
  energy landscape.
\newblock The Journal of Chemical Physics \textbf{132}(23), 234115 (2010)

\bibitem{weinberger}
Niyogi, P., Smale, S., Weinberger, S.: Finding the homology of submanifolds
  with high confidence from random samples.
\newblock Discrete \& Computational Geometry \textbf{39}, 419--441 (2008)

\bibitem{roadmap}
Otter, N., Porter, M.A., Tillmann, U., Grindrod, P., Harrington, H.A.: A
  roadmap for the computation of persistent homology.
\newblock European Physical Journal --- Data Science \textbf{6}, 17 (2017)

\bibitem{kde_submanifold}
Ozakin, A., Gray, A.: Submanifold density estimation.
\newblock In: Y.~Bengio, D.~Schuurmans, J.~Lafferty, C.~Williams, A.~Culotta
  (eds.) Advances in Neural Information Processing Systems, vol.~22. Curran
  Associates, Inc. (2009)

\bibitem{Rman_stats}
Pennec, X.: Intrinsic statistics on riemannian manifolds: Basic tools for
  geometric measurements.
\newblock Journal of Mathematical Imaging and Vision \textbf{25}, 127 (2006)

\bibitem{petersen}
Petersen, P.: Riemannian Geometry, \emph{Graduate Texts in Mathematics}, vol.
  171, 2nd edn.
\newblock Springer, New York, NY (2006)

\bibitem{LLE}
Roweis, S.T., Saul, L.K.: Nonlinear dimensionality reduction by locally linear
  embedding.
\newblock Science \textbf{290}(5500), 2323--2326 (2000)

\bibitem{scotts_rule}
Scott, D.W.: Multivariate Density Estimation: Theory, Practice, and
  Visualization.
\newblock Wiley Series in Probability and Statistics. John Wiley \& Sons, Inc.
  (1992)

\bibitem{Cisomap}
Silva, V.d., Tenenbaum, J.B.: Global versus local methods in nonlinear
  dimensionality reduction.
\newblock In: Proceedings of the 15th International Conference on Neural
  Information Processing Systems, NIPS'02, pp. 721--728. MIT Press, Cambridge,
  MA, USA (2002)

\bibitem{kernels}
Silverman, B.W.: Density Estimation for Statistics and Data Analysis,
  \emph{Monographs on Statistics and Applied Probability}, vol.~26, 1st edn.
\newblock Chapman Hall, London (1986)

\bibitem{stolz2017}
Stolz, B.J., Harrington, H.A., Porter, M.A.: Persistent homology of
  time-dependent functional networks constructed from coupled time series.
\newblock Chaos \textbf{27}, 047410 (2017)

\bibitem{top_regimes}
Strommen, K., Chantry, M., Dorrington, J., Otter, N.: A topological perspective
  on weather regimes.
\newblock arXiv:2104.03196v3  (2021)

\bibitem{takens}
Takens, F.: Detecting strange attractors in turbulence.
\newblock In: D.~Rand, L.S. Young (eds.) Dynamical Systems and Turbulence,
  Lecture Notes in Mathematics, vol. 898, pp. 366--381. Springer, Berlin,
  Heidelberg (1981)

\bibitem{isomap}
Tenenbaum, J.B., {de Silva}, V., Langford, J.C.: A global geometric framework
  for nonlinear dimensionality reduction.
\newblock Science \textbf{290}(5500), 2319--2323 (2000)

\bibitem{topaz}
Topaz, C.M., Ziegelmeier, L., Halverson, T.: Topological data analysis of
  biological aggregation models.
\newblock PLoS ONE \textbf{10}(5), e0126383 (2015)

\bibitem{continuous_knn}
Tyrus~Berry, T.S.: Consistent manifold representation for topological data
  analysis.
\newblock Foundations of Data Science \textbf{1}(1), 1--38 (2019)

\bibitem{scipy}
Virtanen, P., Gommers, R., Oliphant, T.E., Haberland, M., Reddy, T.,
  Cournapeau, D., Burovski, E., Peterson, P., Weckesser, W., Bright, J., {van
  der Walt}, S.J., Brett, M., Wilson, J., Millman, K.J., Mayorov, N., Nelson,
  A.R.J., Jones, E., Kern, R., Larson, E., Carey, C.J., Polat, {\.I}., Feng,
  Y., Moore, E.W., {VanderPlas}, J., Laxalde, D., Perktold, J., Cimrman, R.,
  Henriksen, I., Quintero, E.A., Harris, C.R., Archibald, A.M., Ribeiro, A.H.,
  Pedregosa, F., {van Mulbregt}, P., {SciPy 1.0 Contributors}: {{SciPy} 1.0:
  Fundamental Algorithms for Scientific Computing in Python}.
\newblock Nature Methods \textbf{17}, 261--272 (2020)

\bibitem{conical_dimestimate}
Yang, X., Michea, S., Zha, H.: Conical dimension as an intrinsic dimension
  estimator and its applications.
\newblock In: Proceedings of the 7th SIAM International Conference on Data
  Mining. Minneapolis, Minnesota (2007)

\bibitem{cassini_book}
Yates, R.C.: A Handbook on Curves and their Properties, p.~8.
\newblock J. W. Edwards, Ann Arbor, MI (1947)

\bibitem{fundthm}
Zomorodian, A., Carlsson, G.: Computing persistent homology.
\newblock Discrete \& Computational Geometry \textbf{33}, 249--274 (2005)

\end{thebibliography}


\end{document}